\documentclass[11pt,a4paper]{amsart}
\usepackage[foot]{amsaddr}
\usepackage{fullpage}

\usepackage[T1]{fontenc}
\usepackage[utf8]{inputenc}
\usepackage[british]{babel}

\usepackage{enumerate}

\usepackage{url}
\usepackage{hyperref}

\usepackage{xcolor}
\definecolor{lgreen}{rgb}{0.0, 0.48, 0.0}
\definecolor{lpurple}{rgb}{0.48, 0.0, 0.48}
\definecolor{bblue}{rgb}{0.2, 0.4, 0.8}
\hypersetup{linktocpage,
            colorlinks=true,
            linkcolor=lgreen,
            citecolor=lpurple,
            linktoc=true}

\usepackage{graphics, graphicx}
\usepackage{caption, subcaption}
\usepackage{tikz,pgfplots}

\usepackage{tikz, tikz-qtree, xcolor}
\usetikzlibrary{decorations.markings}
\usetikzlibrary{shapes.geometric}
\usetikzlibrary{matrix,arrows}

\usepackage{eqnarray,amstext}
\usepackage{qsymbols,xparse}
\usepackage{mathtools}
\usepackage{mathrsfs}
\usepackage{stmaryrd}

\usepackage{subcaption}

\theoremstyle{definition}
\newtheorem{theorem}{Theorem}
\numberwithin{theorem}{section}

\newtheorem{lemma}[theorem]{Lemma}
\newtheorem{prop}[theorem]{Proposition}
\newtheorem{definition}[theorem]{Definition}

\newtheorem{example}[theorem]{Example}
\newtheorem{remark}[theorem]{Remark}

\newcommand\cref\autoref
\newcommand{\lref}[1]{\hyperlink{#1}{Line~\ref*{#1}}}

\let\emptyset\varnothing

\usepackage{xspace}

\usepackage[draft]{minted}

\makeatletter
\def\PYG@reset{\let\PYG@it=\relax \let\PYG@bf=\relax%
    \let\PYG@ul=\relax \let\PYG@tc=\relax%
    \let\PYG@bc=\relax \let\PYG@ff=\relax}
\def\PYG@tok#1{\csname PYG@tok@#1\endcsname}
\def\PYG@toks#1+{\ifx\relax#1\empty\else%
    \PYG@tok{#1}\expandafter\PYG@toks\fi}
\def\PYG@do#1{\PYG@bc{\PYG@tc{\PYG@ul{%
    \PYG@it{\PYG@bf{\PYG@ff{#1}}}}}}}
\def\PYG#1#2{\PYG@reset\PYG@toks#1+\relax+\PYG@do{#2}}

\expandafter\def\csname PYG@tok@gd\endcsname{\def\PYG@tc##1{\textcolor[rgb]{0.63,0.00,0.00}{##1}}}
\expandafter\def\csname PYG@tok@gu\endcsname{\let\PYG@bf=\textbf\def\PYG@tc##1{\textcolor[rgb]{0.50,0.00,0.50}{##1}}}
\expandafter\def\csname PYG@tok@gt\endcsname{\def\PYG@tc##1{\textcolor[rgb]{0.00,0.27,0.87}{##1}}}
\expandafter\def\csname PYG@tok@gs\endcsname{\let\PYG@bf=\textbf}
\expandafter\def\csname PYG@tok@gr\endcsname{\def\PYG@tc##1{\textcolor[rgb]{1.00,0.00,0.00}{##1}}}
\expandafter\def\csname PYG@tok@cm\endcsname{\let\PYG@it=\textit\def\PYG@tc##1{\textcolor[rgb]{0.25,0.50,0.50}{##1}}}
\expandafter\def\csname PYG@tok@vg\endcsname{\def\PYG@tc##1{\textcolor[rgb]{0.10,0.09,0.49}{##1}}}
\expandafter\def\csname PYG@tok@vi\endcsname{\def\PYG@tc##1{\textcolor[rgb]{0.10,0.09,0.49}{##1}}}
\expandafter\def\csname PYG@tok@vm\endcsname{\def\PYG@tc##1{\textcolor[rgb]{0.10,0.09,0.49}{##1}}}
\expandafter\def\csname PYG@tok@mh\endcsname{\def\PYG@tc##1{\textcolor[rgb]{0.40,0.40,0.40}{##1}}}
\expandafter\def\csname PYG@tok@cs\endcsname{\let\PYG@it=\textit\def\PYG@tc##1{\textcolor[rgb]{0.25,0.50,0.50}{##1}}}
\expandafter\def\csname PYG@tok@ge\endcsname{\let\PYG@it=\textit}
\expandafter\def\csname PYG@tok@vc\endcsname{\def\PYG@tc##1{\textcolor[rgb]{0.10,0.09,0.49}{##1}}}
\expandafter\def\csname PYG@tok@il\endcsname{\def\PYG@tc##1{\textcolor[rgb]{0.40,0.40,0.40}{##1}}}
\expandafter\def\csname PYG@tok@go\endcsname{\def\PYG@tc##1{\textcolor[rgb]{0.53,0.53,0.53}{##1}}}
\expandafter\def\csname PYG@tok@cp\endcsname{\def\PYG@tc##1{\textcolor[rgb]{0.74,0.48,0.00}{##1}}}
\expandafter\def\csname PYG@tok@gi\endcsname{\def\PYG@tc##1{\textcolor[rgb]{0.00,0.63,0.00}{##1}}}
\expandafter\def\csname PYG@tok@gh\endcsname{\let\PYG@bf=\textbf\def\PYG@tc##1{\textcolor[rgb]{0.00,0.00,0.50}{##1}}}
\expandafter\def\csname PYG@tok@ni\endcsname{\let\PYG@bf=\textbf\def\PYG@tc##1{\textcolor[rgb]{0.60,0.60,0.60}{##1}}}
\expandafter\def\csname PYG@tok@nl\endcsname{\def\PYG@tc##1{\textcolor[rgb]{0.63,0.63,0.00}{##1}}}
\expandafter\def\csname PYG@tok@nn\endcsname{\let\PYG@bf=\textbf\def\PYG@tc##1{\textcolor[rgb]{0.00,0.00,1.00}{##1}}}
\expandafter\def\csname PYG@tok@no\endcsname{\def\PYG@tc##1{\textcolor[rgb]{0.53,0.00,0.00}{##1}}}
\expandafter\def\csname PYG@tok@na\endcsname{\def\PYG@tc##1{\textcolor[rgb]{0.49,0.56,0.16}{##1}}}
\expandafter\def\csname PYG@tok@nb\endcsname{\def\PYG@tc##1{\textcolor[rgb]{0.00,0.50,0.00}{##1}}}
\expandafter\def\csname PYG@tok@nc\endcsname{\let\PYG@bf=\textbf\def\PYG@tc##1{\textcolor[rgb]{0.00,0.00,1.00}{##1}}}
\expandafter\def\csname PYG@tok@nd\endcsname{\def\PYG@tc##1{\textcolor[rgb]{0.67,0.13,1.00}{##1}}}
\expandafter\def\csname PYG@tok@ne\endcsname{\let\PYG@bf=\textbf\def\PYG@tc##1{\textcolor[rgb]{0.82,0.25,0.23}{##1}}}
\expandafter\def\csname PYG@tok@nf\endcsname{\def\PYG@tc##1{\textcolor[rgb]{0.00,0.00,1.00}{##1}}}
\expandafter\def\csname PYG@tok@si\endcsname{\let\PYG@bf=\textbf\def\PYG@tc##1{\textcolor[rgb]{0.73,0.40,0.53}{##1}}}
\expandafter\def\csname PYG@tok@s2\endcsname{\def\PYG@tc##1{\textcolor[rgb]{0.73,0.13,0.13}{##1}}}
\expandafter\def\csname PYG@tok@nt\endcsname{\let\PYG@bf=\textbf\def\PYG@tc##1{\textcolor[rgb]{0.00,0.50,0.00}{##1}}}
\expandafter\def\csname PYG@tok@nv\endcsname{\def\PYG@tc##1{\textcolor[rgb]{0.10,0.09,0.49}{##1}}}
\expandafter\def\csname PYG@tok@s1\endcsname{\def\PYG@tc##1{\textcolor[rgb]{0.73,0.13,0.13}{##1}}}
\expandafter\def\csname PYG@tok@dl\endcsname{\def\PYG@tc##1{\textcolor[rgb]{0.73,0.13,0.13}{##1}}}
\expandafter\def\csname PYG@tok@ch\endcsname{\let\PYG@it=\textit\def\PYG@tc##1{\textcolor[rgb]{0.25,0.50,0.50}{##1}}}
\expandafter\def\csname PYG@tok@m\endcsname{\def\PYG@tc##1{\textcolor[rgb]{0.40,0.40,0.40}{##1}}}
\expandafter\def\csname PYG@tok@gp\endcsname{\let\PYG@bf=\textbf\def\PYG@tc##1{\textcolor[rgb]{0.00,0.00,0.50}{##1}}}
\expandafter\def\csname PYG@tok@sh\endcsname{\def\PYG@tc##1{\textcolor[rgb]{0.73,0.13,0.13}{##1}}}
\expandafter\def\csname PYG@tok@ow\endcsname{\let\PYG@bf=\textbf\def\PYG@tc##1{\textcolor[rgb]{0.67,0.13,1.00}{##1}}}
\expandafter\def\csname PYG@tok@sx\endcsname{\def\PYG@tc##1{\textcolor[rgb]{0.00,0.50,0.00}{##1}}}
\expandafter\def\csname PYG@tok@bp\endcsname{\def\PYG@tc##1{\textcolor[rgb]{0.00,0.50,0.00}{##1}}}
\expandafter\def\csname PYG@tok@c1\endcsname{\let\PYG@it=\textit\def\PYG@tc##1{\textcolor[rgb]{0.25,0.50,0.50}{##1}}}
\expandafter\def\csname PYG@tok@fm\endcsname{\def\PYG@tc##1{\textcolor[rgb]{0.00,0.00,1.00}{##1}}}
\expandafter\def\csname PYG@tok@o\endcsname{\def\PYG@tc##1{\textcolor[rgb]{0.40,0.40,0.40}{##1}}}
\expandafter\def\csname PYG@tok@kc\endcsname{\let\PYG@bf=\textbf\def\PYG@tc##1{\textcolor[rgb]{0.00,0.50,0.00}{##1}}}
\expandafter\def\csname PYG@tok@c\endcsname{\let\PYG@it=\textit\def\PYG@tc##1{\textcolor[rgb]{0.25,0.50,0.50}{##1}}}
\expandafter\def\csname PYG@tok@mf\endcsname{\def\PYG@tc##1{\textcolor[rgb]{0.40,0.40,0.40}{##1}}}
\expandafter\def\csname PYG@tok@err\endcsname{\def\PYG@bc##1{\setlength{\fboxsep}{0pt}\fcolorbox[rgb]{1.00,0.00,0.00}{1,1,1}{\strut ##1}}}
\expandafter\def\csname PYG@tok@mb\endcsname{\def\PYG@tc##1{\textcolor[rgb]{0.40,0.40,0.40}{##1}}}
\expandafter\def\csname PYG@tok@ss\endcsname{\def\PYG@tc##1{\textcolor[rgb]{0.10,0.09,0.49}{##1}}}
\expandafter\def\csname PYG@tok@sr\endcsname{\def\PYG@tc##1{\textcolor[rgb]{0.73,0.40,0.53}{##1}}}
\expandafter\def\csname PYG@tok@mo\endcsname{\def\PYG@tc##1{\textcolor[rgb]{0.40,0.40,0.40}{##1}}}
\expandafter\def\csname PYG@tok@kd\endcsname{\let\PYG@bf=\textbf\def\PYG@tc##1{\textcolor[rgb]{0.00,0.50,0.00}{##1}}}
\expandafter\def\csname PYG@tok@mi\endcsname{\def\PYG@tc##1{\textcolor[rgb]{0.40,0.40,0.40}{##1}}}
\expandafter\def\csname PYG@tok@kn\endcsname{\let\PYG@bf=\textbf\def\PYG@tc##1{\textcolor[rgb]{0.00,0.50,0.00}{##1}}}
\expandafter\def\csname PYG@tok@cpf\endcsname{\let\PYG@it=\textit\def\PYG@tc##1{\textcolor[rgb]{0.25,0.50,0.50}{##1}}}
\expandafter\def\csname PYG@tok@kr\endcsname{\let\PYG@bf=\textbf\def\PYG@tc##1{\textcolor[rgb]{0.00,0.50,0.00}{##1}}}
\expandafter\def\csname PYG@tok@s\endcsname{\def\PYG@tc##1{\textcolor[rgb]{0.73,0.13,0.13}{##1}}}
\expandafter\def\csname PYG@tok@kp\endcsname{\def\PYG@tc##1{\textcolor[rgb]{0.00,0.50,0.00}{##1}}}
\expandafter\def\csname PYG@tok@w\endcsname{\def\PYG@tc##1{\textcolor[rgb]{0.73,0.73,0.73}{##1}}}
\expandafter\def\csname PYG@tok@kt\endcsname{\def\PYG@tc##1{\textcolor[rgb]{0.69,0.00,0.25}{##1}}}
\expandafter\def\csname PYG@tok@sc\endcsname{\def\PYG@tc##1{\textcolor[rgb]{0.73,0.13,0.13}{##1}}}
\expandafter\def\csname PYG@tok@sb\endcsname{\def\PYG@tc##1{\textcolor[rgb]{0.73,0.13,0.13}{##1}}}
\expandafter\def\csname PYG@tok@sa\endcsname{\def\PYG@tc##1{\textcolor[rgb]{0.73,0.13,0.13}{##1}}}
\expandafter\def\csname PYG@tok@k\endcsname{\let\PYG@bf=\textbf\def\PYG@tc##1{\textcolor[rgb]{0.00,0.50,0.00}{##1}}}
\expandafter\def\csname PYG@tok@se\endcsname{\let\PYG@bf=\textbf\def\PYG@tc##1{\textcolor[rgb]{0.73,0.40,0.13}{##1}}}
\expandafter\def\csname PYG@tok@sd\endcsname{\let\PYG@it=\textit\def\PYG@tc##1{\textcolor[rgb]{0.73,0.13,0.13}{##1}}}

\makeatother

\makeatletter
\def\PYGdefault@reset{\let\PYGdefault@it=\relax \let\PYGdefault@bf=\relax%
    \let\PYGdefault@ul=\relax \let\PYGdefault@tc=\relax%
    \let\PYGdefault@bc=\relax \let\PYGdefault@ff=\relax}
\def\PYGdefault@tok#1{\csname PYGdefault@tok@#1\endcsname}
\def\PYGdefault@toks#1+{\ifx\relax#1\empty\else%
    \PYGdefault@tok{#1}\expandafter\PYGdefault@toks\fi}
\def\PYGdefault@do#1{\PYGdefault@bc{\PYGdefault@tc{\PYGdefault@ul{%
    \PYGdefault@it{\PYGdefault@bf{\PYGdefault@ff{#1}}}}}}}
\def\PYGdefault#1#2{\PYGdefault@reset\PYGdefault@toks#1+\relax+\PYGdefault@do{#2}}

\expandafter\def\csname PYGdefault@tok@gd\endcsname{\def\PYGdefault@tc##1{\textcolor[rgb]{0.63,0.00,0.00}{##1}}}
\expandafter\def\csname PYGdefault@tok@gu\endcsname{\let\PYGdefault@bf=\textbf\def\PYGdefault@tc##1{\textcolor[rgb]{0.50,0.00,0.50}{##1}}}
\expandafter\def\csname PYGdefault@tok@gt\endcsname{\def\PYGdefault@tc##1{\textcolor[rgb]{0.00,0.27,0.87}{##1}}}
\expandafter\def\csname PYGdefault@tok@gs\endcsname{\let\PYGdefault@bf=\textbf}
\expandafter\def\csname PYGdefault@tok@gr\endcsname{\def\PYGdefault@tc##1{\textcolor[rgb]{1.00,0.00,0.00}{##1}}}
\expandafter\def\csname PYGdefault@tok@cm\endcsname{\let\PYGdefault@it=\textit\def\PYGdefault@tc##1{\textcolor[rgb]{0.25,0.50,0.50}{##1}}}
\expandafter\def\csname PYGdefault@tok@vg\endcsname{\def\PYGdefault@tc##1{\textcolor[rgb]{0.10,0.09,0.49}{##1}}}
\expandafter\def\csname PYGdefault@tok@vi\endcsname{\def\PYGdefault@tc##1{\textcolor[rgb]{0.10,0.09,0.49}{##1}}}
\expandafter\def\csname PYGdefault@tok@vm\endcsname{\def\PYGdefault@tc##1{\textcolor[rgb]{0.10,0.09,0.49}{##1}}}
\expandafter\def\csname PYGdefault@tok@mh\endcsname{\def\PYGdefault@tc##1{\textcolor[rgb]{0.40,0.40,0.40}{##1}}}
\expandafter\def\csname PYGdefault@tok@cs\endcsname{\let\PYGdefault@it=\textit\def\PYGdefault@tc##1{\textcolor[rgb]{0.25,0.50,0.50}{##1}}}
\expandafter\def\csname PYGdefault@tok@ge\endcsname{\let\PYGdefault@it=\textit}
\expandafter\def\csname PYGdefault@tok@vc\endcsname{\def\PYGdefault@tc##1{\textcolor[rgb]{0.10,0.09,0.49}{##1}}}
\expandafter\def\csname PYGdefault@tok@il\endcsname{\def\PYGdefault@tc##1{\textcolor[rgb]{0.40,0.40,0.40}{##1}}}
\expandafter\def\csname PYGdefault@tok@go\endcsname{\def\PYGdefault@tc##1{\textcolor[rgb]{0.53,0.53,0.53}{##1}}}
\expandafter\def\csname PYGdefault@tok@cp\endcsname{\def\PYGdefault@tc##1{\textcolor[rgb]{0.74,0.48,0.00}{##1}}}
\expandafter\def\csname PYGdefault@tok@gi\endcsname{\def\PYGdefault@tc##1{\textcolor[rgb]{0.00,0.63,0.00}{##1}}}
\expandafter\def\csname PYGdefault@tok@gh\endcsname{\let\PYGdefault@bf=\textbf\def\PYGdefault@tc##1{\textcolor[rgb]{0.00,0.00,0.50}{##1}}}
\expandafter\def\csname PYGdefault@tok@ni\endcsname{\let\PYGdefault@bf=\textbf\def\PYGdefault@tc##1{\textcolor[rgb]{0.60,0.60,0.60}{##1}}}
\expandafter\def\csname PYGdefault@tok@nl\endcsname{\def\PYGdefault@tc##1{\textcolor[rgb]{0.63,0.63,0.00}{##1}}}
\expandafter\def\csname PYGdefault@tok@nn\endcsname{\let\PYGdefault@bf=\textbf\def\PYGdefault@tc##1{\textcolor[rgb]{0.00,0.00,1.00}{##1}}}
\expandafter\def\csname PYGdefault@tok@no\endcsname{\def\PYGdefault@tc##1{\textcolor[rgb]{0.53,0.00,0.00}{##1}}}
\expandafter\def\csname PYGdefault@tok@na\endcsname{\def\PYGdefault@tc##1{\textcolor[rgb]{0.49,0.56,0.16}{##1}}}
\expandafter\def\csname PYGdefault@tok@nb\endcsname{\def\PYGdefault@tc##1{\textcolor[rgb]{0.00,0.50,0.00}{##1}}}
\expandafter\def\csname PYGdefault@tok@nc\endcsname{\let\PYGdefault@bf=\textbf\def\PYGdefault@tc##1{\textcolor[rgb]{0.00,0.00,1.00}{##1}}}
\expandafter\def\csname PYGdefault@tok@nd\endcsname{\def\PYGdefault@tc##1{\textcolor[rgb]{0.67,0.13,1.00}{##1}}}
\expandafter\def\csname PYGdefault@tok@ne\endcsname{\let\PYGdefault@bf=\textbf\def\PYGdefault@tc##1{\textcolor[rgb]{0.82,0.25,0.23}{##1}}}
\expandafter\def\csname PYGdefault@tok@nf\endcsname{\def\PYGdefault@tc##1{\textcolor[rgb]{0.00,0.00,1.00}{##1}}}
\expandafter\def\csname PYGdefault@tok@si\endcsname{\let\PYGdefault@bf=\textbf\def\PYGdefault@tc##1{\textcolor[rgb]{0.73,0.40,0.53}{##1}}}
\expandafter\def\csname PYGdefault@tok@s2\endcsname{\def\PYGdefault@tc##1{\textcolor[rgb]{0.73,0.13,0.13}{##1}}}
\expandafter\def\csname PYGdefault@tok@nt\endcsname{\let\PYGdefault@bf=\textbf\def\PYGdefault@tc##1{\textcolor[rgb]{0.00,0.50,0.00}{##1}}}
\expandafter\def\csname PYGdefault@tok@nv\endcsname{\def\PYGdefault@tc##1{\textcolor[rgb]{0.10,0.09,0.49}{##1}}}
\expandafter\def\csname PYGdefault@tok@s1\endcsname{\def\PYGdefault@tc##1{\textcolor[rgb]{0.73,0.13,0.13}{##1}}}
\expandafter\def\csname PYGdefault@tok@dl\endcsname{\def\PYGdefault@tc##1{\textcolor[rgb]{0.73,0.13,0.13}{##1}}}
\expandafter\def\csname PYGdefault@tok@ch\endcsname{\let\PYGdefault@it=\textit\def\PYGdefault@tc##1{\textcolor[rgb]{0.25,0.50,0.50}{##1}}}
\expandafter\def\csname PYGdefault@tok@m\endcsname{\def\PYGdefault@tc##1{\textcolor[rgb]{0.40,0.40,0.40}{##1}}}
\expandafter\def\csname PYGdefault@tok@gp\endcsname{\let\PYGdefault@bf=\textbf\def\PYGdefault@tc##1{\textcolor[rgb]{0.00,0.00,0.50}{##1}}}
\expandafter\def\csname PYGdefault@tok@sh\endcsname{\def\PYGdefault@tc##1{\textcolor[rgb]{0.73,0.13,0.13}{##1}}}
\expandafter\def\csname PYGdefault@tok@ow\endcsname{\let\PYGdefault@bf=\textbf\def\PYGdefault@tc##1{\textcolor[rgb]{0.67,0.13,1.00}{##1}}}
\expandafter\def\csname PYGdefault@tok@sx\endcsname{\def\PYGdefault@tc##1{\textcolor[rgb]{0.00,0.50,0.00}{##1}}}
\expandafter\def\csname PYGdefault@tok@bp\endcsname{\def\PYGdefault@tc##1{\textcolor[rgb]{0.00,0.50,0.00}{##1}}}
\expandafter\def\csname PYGdefault@tok@c1\endcsname{\let\PYGdefault@it=\textit\def\PYGdefault@tc##1{\textcolor[rgb]{0.25,0.50,0.50}{##1}}}
\expandafter\def\csname PYGdefault@tok@fm\endcsname{\def\PYGdefault@tc##1{\textcolor[rgb]{0.00,0.00,1.00}{##1}}}
\expandafter\def\csname PYGdefault@tok@o\endcsname{\def\PYGdefault@tc##1{\textcolor[rgb]{0.40,0.40,0.40}{##1}}}
\expandafter\def\csname PYGdefault@tok@kc\endcsname{\let\PYGdefault@bf=\textbf\def\PYGdefault@tc##1{\textcolor[rgb]{0.00,0.50,0.00}{##1}}}
\expandafter\def\csname PYGdefault@tok@c\endcsname{\let\PYGdefault@it=\textit\def\PYGdefault@tc##1{\textcolor[rgb]{0.25,0.50,0.50}{##1}}}
\expandafter\def\csname PYGdefault@tok@mf\endcsname{\def\PYGdefault@tc##1{\textcolor[rgb]{0.40,0.40,0.40}{##1}}}
\expandafter\def\csname PYGdefault@tok@err\endcsname{\def\PYGdefault@bc##1{\setlength{\fboxsep}{0pt}\fcolorbox[rgb]{1.00,0.00,0.00}{1,1,1}{\strut ##1}}}
\expandafter\def\csname PYGdefault@tok@mb\endcsname{\def\PYGdefault@tc##1{\textcolor[rgb]{0.40,0.40,0.40}{##1}}}
\expandafter\def\csname PYGdefault@tok@ss\endcsname{\def\PYGdefault@tc##1{\textcolor[rgb]{0.10,0.09,0.49}{##1}}}
\expandafter\def\csname PYGdefault@tok@sr\endcsname{\def\PYGdefault@tc##1{\textcolor[rgb]{0.73,0.40,0.53}{##1}}}
\expandafter\def\csname PYGdefault@tok@mo\endcsname{\def\PYGdefault@tc##1{\textcolor[rgb]{0.40,0.40,0.40}{##1}}}
\expandafter\def\csname PYGdefault@tok@kd\endcsname{\let\PYGdefault@bf=\textbf\def\PYGdefault@tc##1{\textcolor[rgb]{0.00,0.50,0.00}{##1}}}
\expandafter\def\csname PYGdefault@tok@mi\endcsname{\def\PYGdefault@tc##1{\textcolor[rgb]{0.40,0.40,0.40}{##1}}}
\expandafter\def\csname PYGdefault@tok@kn\endcsname{\let\PYGdefault@bf=\textbf\def\PYGdefault@tc##1{\textcolor[rgb]{0.00,0.50,0.00}{##1}}}
\expandafter\def\csname PYGdefault@tok@cpf\endcsname{\let\PYGdefault@it=\textit\def\PYGdefault@tc##1{\textcolor[rgb]{0.25,0.50,0.50}{##1}}}
\expandafter\def\csname PYGdefault@tok@kr\endcsname{\let\PYGdefault@bf=\textbf\def\PYGdefault@tc##1{\textcolor[rgb]{0.00,0.50,0.00}{##1}}}
\expandafter\def\csname PYGdefault@tok@s\endcsname{\def\PYGdefault@tc##1{\textcolor[rgb]{0.73,0.13,0.13}{##1}}}
\expandafter\def\csname PYGdefault@tok@kp\endcsname{\def\PYGdefault@tc##1{\textcolor[rgb]{0.00,0.50,0.00}{##1}}}
\expandafter\def\csname PYGdefault@tok@w\endcsname{\def\PYGdefault@tc##1{\textcolor[rgb]{0.73,0.73,0.73}{##1}}}
\expandafter\def\csname PYGdefault@tok@kt\endcsname{\def\PYGdefault@tc##1{\textcolor[rgb]{0.69,0.00,0.25}{##1}}}
\expandafter\def\csname PYGdefault@tok@sc\endcsname{\def\PYGdefault@tc##1{\textcolor[rgb]{0.73,0.13,0.13}{##1}}}
\expandafter\def\csname PYGdefault@tok@sb\endcsname{\def\PYGdefault@tc##1{\textcolor[rgb]{0.73,0.13,0.13}{##1}}}
\expandafter\def\csname PYGdefault@tok@sa\endcsname{\def\PYGdefault@tc##1{\textcolor[rgb]{0.73,0.13,0.13}{##1}}}
\expandafter\def\csname PYGdefault@tok@k\endcsname{\let\PYGdefault@bf=\textbf\def\PYGdefault@tc##1{\textcolor[rgb]{0.00,0.50,0.00}{##1}}}
\expandafter\def\csname PYGdefault@tok@se\endcsname{\let\PYGdefault@bf=\textbf\def\PYGdefault@tc##1{\textcolor[rgb]{0.73,0.40,0.13}{##1}}}
\expandafter\def\csname PYGdefault@tok@sd\endcsname{\let\PYGdefault@it=\textit\def\PYGdefault@tc##1{\textcolor[rgb]{0.73,0.13,0.13}{##1}}}

\makeatother

\newcommand{\nats}{\mathbb{N}}
\newcommand{\set}[1]{\{#1\}}
\newcommand{\seq}[1]{\left(#1\right)}
\newcommand{\idx}[1]{\mbox{\underline{\sf #1}}}

\newcommand{\luterm}{$`l`y$\nobreakdash-term\xspace}
\newcommand{\luterms}{$`l`y$\nobreakdash-terms\xspace}

\newcommand{\lterms}{$`l$\nobreakdash-terms\xspace}
\newcommand{\lcalculus}{$`l$\nobreakdash-calculus\xspace}
\newcommand{\lucalculus}{$`l`y$\nobreakdash-calculus\xspace}

\newcommand{\fip}{{\sf FIP}\xspace}
\newcommand{\urg}{$`y$\nobreakdash-{\sf RG}\xspace}

\let\succ\relax
\DeclareMathOperator*{\succ}{\mbox{\sf S}}
\DeclareMathOperator*{\arity}{\mbox{\sf arity}}

\begin{document}
\renewcommand{\sectionautorefname}{Section}

\title{Towards the average-case analysis\\of substitution resolution in
$\lambda$-calculus}
\author{Maciej Bendkowski}
\address{Theoretical Computer Science Department\\
  Faculty of Mathematics and Computer Science\\
  Jagiellonian University\\
  ul. Prof. {\L}ojasiewicza 6, 30-348 Krak\'ow, Poland}
\email{maciej.bendkowski@tcs.uj.edu.pl}
\date{\today}

\begin{abstract}
    Substitution resolution supports the computational character of
    $`b$\nobreakdash-reduction, complementing its execution with a
    capture-avoiding exchange of terms for bound variables. Alas, the
    meta-level definition of substitution, masking a non-trivial
    computation, turns $`b$\nobreakdash-reduction into an atomic
    rewriting rule, despite its varying operational complexity.

    In the current paper we propose a somewhat indirect average-case analysis of
    substitution resolution in the classic \lcalculus, based on the quantitative
    analysis of substitution in $`l`y$, an extension of \lcalculus
    internalising the $`y$\nobreakdash-calculus of explicit substitutions.
    Within this framework, we show that for any fixed $n \geq 0$, the
    probability that a uniformly random, conditioned on size, \luterm
    $`y$\nobreakdash-normalises in $n$ normal-order (i.e.~leftmost-outermost)
    reduction steps tends to a computable limit as the term size tends to infinity.

    For that purpose, we establish an effective hierarchy
    $\seq{\mathscr{G}_n}_n$ of regular tree grammars partitioning
    $`y$\nobreakdash-normalisable terms into classes of terms normalising in $n$
    normal-order rewriting steps. The main technical ingredient in our
    construction is an inductive approach to the construction of
    $\mathscr{G}_{n+1}$ out of $\mathscr{G}_n$ based, in turn, on the
    algorithmic construction of finite intersection partitions, inspired by
    Robinson's unification algorithm.

    Finally, we briefly discuss applications of our approach to other term
    rewriting systems, focusing on two closely related formalisms, i.e.~the full
    \lucalculus and combinatory logic.
\end{abstract}

\maketitle

\section{Introduction}
Traditional, machine-based computational models, such as Turing machines or
RAMs, admit a natural notion of an atomic computation step, closely reflecting
the actual operational cost of executing the represented computations.
Unfortunately, this is not quite the case for computational models based on term
rewriting systems with substitution, such as the classic \lcalculus.  Given the
(traditionally) epitheoretic nature of substitution, the single rewriting rule
of $`b$\nobreakdash-reduction $(`l x. a)b \to_\beta a[x := b]$ masks a
non-trivial computation of resolving (i.e.~executing) the pending substitution
of $b$ for occurrences of $x$ in $a$.  Moreover, unlike machine-based models,
\lcalculus (as other term rewriting systems) does not impose a strict,
deterministic evaluation mechanism.  Consequently, various strategies for
resolving substitutions can be used, even more obfuscating the operational
semantics of $`b$\nobreakdash-reduction and hence also its operational cost.
Those subtle nuances hidden behind the implementation details of substitution
resolution are in fact one of the core issues in establishing reasonable cost
models for the classic \lcalculus, relating it with other, machine-based
computational models, see~\cite{Lawall:1996:OII:232627.232639}.

In order to resolve this apparent inadequacy, Abadi et al.~proposed to refine
substitution in the classic \lcalculus and decompose it into a series of atomic
rewriting steps, internalising in effect the calculus of executing
substitutions~\cite{abadi1991}. Substitutions become first-class citizens and,
in effect, can be manipulated together with regular terms.  Consequently, the
general framework of explicit substitutions provides a machine-independent setup
for the operational semantics of substitution, based on a finite set of unit
rewriting primitives. Remarkably, with the help of linear substitution calculus
(a resource aware calculus of explicit substitutions) Accattoli and dal Lago
showed recently that the leftmost-outermost $`b$\nobreakdash-reduction strategy
is a reasonable invariant cost model for \lcalculus, and hence it is able to
simulate RAMs (or equivalent, machine-based models) within a polynomial time
overhead~\cite{DBLP:journals/corr/AccattoliL16}.

Various subtleties of substitution resolution, reflected in the variety of
available calculi of explicit substitutions, induce different operational
semantics for executing substitutions in \lcalculus. This abundance of
approaches is perhaps the main barrier in establishing a systematic,
quantitative analysis of the operational complexity of substitution resolution
and, among other things, a term rewriting analogue of classic average-case
complexity analysis.  In the current paper we propose a step towards filling
this gap by offering a quantitative approach to substitution resolution in
Lescanne's \lucalculus of explicit substitutions~\cite{Lescanne1994}. In
particular, we focus on the following, average-case analysis type of question. Having
fixed arbitrary non-negative $n$, what is the probability that a (uniformly) random \luterm of
given size is $`y$\nobreakdash-normalisable (i.e.~can be reduced to a normal form
without explicit substitutions) in exactly $n$ leftmost-outermost reduction
steps? Furthermore, how does this probability distribution change with the term
size tending to infinity?

We address the above questions using a two-step approach. First, we exhibit an
effective (i.e.~computable) hierarchy $\seq{\mathscr{G}_n}_n$ of unambiguous
regular tree grammars with the property that $\mathscr{G}_n$ describes the
language of terms $`y$\nobreakdash-normalising in precisely $n$
leftmost-outermost $`y$\nobreakdash-rewriting steps. Next, borrowing techniques
from analytic combinatorics, we analyse the limit proportion of terms
$`y$\nobreakdash-normalising in $n$ normal-order steps. To that end, we
construct appropriate generating functions and provide asymptotic estimates for
the number of \luterms $`y$\nobreakdash-normalising in $n$ normal-order
reduction steps. As a result, we base our approach on a direct quantitative
analysis of the $`y$ term rewriting system, measuring the operational cost of
evaluating substitution in terms of the number of leftmost-outermost rewriting
steps required to reach a ($`y$-)normal form.

The paper is structured as follows. In~\cref{sec:preliminaries} we outline
\lucalculus and the framework of regular tree grammars, establishing the
necessary terminology for the reminder of the paper. Next,
in~\cref{sec:reduction:grammars}, we prepare the background for the construction
of $\seq{\mathscr{G}_n}_n$. In particular, we sketch its general, intuitive
scheme. In~\cref{sec:finite:intersection:partition} we introduce the main tool
of finite intersection partitions and show that it is indeed constructible in
the context of generated reduction grammars. Afterwards,
in~\cref{sec:the:construction}, we show how finite intersection partitions can
be used in the construction of new productions in $\mathscr{G}_{n+1}$ based on
productions in the grammar $\mathscr{G}_n$. Having constructed
$\seq{\mathscr{G}_n}_n$ we then proceed to the main quantitative analysis of
$`y$\nobreakdash-calculus using methods of analytic combinatorics, see~\cref{sec:analytic:combinatorics}.  Finally,
in~\cref{sec:other:applications} we discuss broader applications of our technique
to other term rewriting systems, based on the examples of \lucalculus and combinatory
logic, and conclude the paper in the final~\cref{sec:conclusion}.

\section{Preliminaries}\label{sec:preliminaries}
\subsection{Lambda upsilon calculus}
$`l`y$ (lambda upsilon) is a simple, first-order term rewriting system extending
the classic \lcalculus based on de~Bruijn indices~\cite{deBruijn1972} with the
calculus of resolving pending substitutions~\cite{Lescanne1994,lescanne96}.  Its
formal terms, so-called $`l`y$\nobreakdash-terms, are comprised of de~Bruijn
indices $\idx{n}$, application, abstraction, together with an additional,
explicit \emph{closure} operator $[\cdot]$ standing for unresolved
substitutions. De~Bruijn indices are represented in unary base expansion. In
other words, $\idx{n}$ is encoded as an $n$\nobreakdash-fold application of the
successor operator $\succ$ to zero~$\idx{0}$. Substitutions, in turn, consist of three
primitives, i.e.~a constant \emph{shift} $\uparrow$, a unary \emph{lift} operator
$\Uparrow$, mapping substitutions onto substitutions, and a unary
\emph{slash} operator $/$, mapping terms onto substitutions. Terms containing
closures are called \emph{impure} whereas terms without them are said to be
\emph{pure}.~\cref{fig:lucalculus} summarises the formal specification of
\luterms~and the corresponding rewriting system $`l`y$.

\begin{figure}[hbt!]
\centering
\begin{subfigure}{.45\textwidth}
    \begin{align}\label{eq:luterms:spec}
    \begin{split}
    t &::= \idx{n} ~|~ `l t ~|~ t t ~|~ t [s]\\
    s &::= t/ ~|~ \Uparrow(s) ~|~ \uparrow\\
    \idx{n} &::= \idx{0} ~|~ \succ \idx{n}.
    \end{split}
\end{align}
    \caption{Terms of \lucalculus.}
\label{fig:lambda:ups:specification}
\end{subfigure}%
\begin{subfigure}{.45\textwidth}
    \begin{align}
    (`l a) b &\to a [b/] \tag{Beta}\label{red:beta}\\
    (a b)[s] &\to a [s] (b [s]) \tag{App}\label{red:app}\\
    (`l a)[s] &\to `l(a[\Uparrow (s)]) \tag{Lambda}\label{red:lambda}\\
    \idx{0} [a/] &\to a \tag{FVar}\label{red:fvar}\\
    (\succ \idx{n}) [a/] &\to \idx{n} \tag{RVar}\label{red:rvar}\\
    \idx{0} [\Uparrow(s)] &\to \idx{0} \tag{FVarLift}\label{red:fvarlift}\\
    (\succ \idx{n}) [\Uparrow(s)] &\to \idx{n}[s][\uparrow] \tag{RVarLift}\label{red:rvarlift}\\
    \idx{n}[\uparrow] &\to \succ \idx{n} \tag{VarShift}\label{red:varshift}.
\end{align}
    \caption{Rewriting rules.}
    \label{fig:lambda:ups:rewriting:system}
\end{subfigure}
\caption{The $`l`y$-calculus rewriting system.}
    \label{fig:lucalculus}
\end{figure}

\begin{example}
    Note that the well-known combinator $K = `lx y. x$ is represented in the
    de~Bruijn notation as $`l `l \idx{1}$.  The reverse application term $`lx y.
    y x$, on the other hand, is represented as $`l`l\idx{0} \idx{1}$.
    Consequently, in a single $`b$\nobreakdash-reduction step, it holds
    $\left(`l`l \idx{0} \idx{1}\right) K \to_\beta `l \left(\idx{0} K\right)$.

    In $`l`y$, however, this single $`b$\nobreakdash-reduction is decomposed
    into a series of small rewriting steps governing both the
    $`b$\nobreakdash-reduction as well as the subsequent substitution
    resolution.  For instance, we have
        \begin{align}
            \begin{split}
                \left(`l`l \idx{0} \idx{1}\right) K &\to
            \left(`l \idx{0} \idx{1}\right)[K/] \to
            \left(`l \left(\idx{0} \idx{1}\right)[\Uparrow(K/)]\right) \to
            `l \left(\idx{0}
                [\Uparrow(K/)]\right)\left(\idx{1}[\Uparrow(K/)]\right)\\
                &\to `l \left( \idx{0}\left(\idx{1}[\Uparrow(K/)]\right) \right) \to
                `l \left(\idx{0}\left(\idx{0}[K/][\uparrow]\right)\right) \to
                `l \left(\idx{0} \left(K [\uparrow]\right)\right).
            \end{split}
        \end{align}
    Furthermore,
        \begin{align}
            \begin{split}
                K [\uparrow] &= \left(`l `l \idx{1} \right)[\uparrow]
                \to `l \left((`l \idx{1}) [\Uparrow(\uparrow)]\right) \to
                `l `l \left( \idx{1} [\Uparrow(\Uparrow(\uparrow))]\right)\\
                &\to `l `l \left(\idx{0} [\Uparrow(\uparrow)][\uparrow] \right)
                \to `l `l \left(\idx{0} [\uparrow]\right)\\
                &\to `l `l \idx{1} = K
            \end{split}
        \end{align}
    hence indeed $\left(`l`l \idx{0} \idx{1}\right) K$ rewrites to $`l
    \left(\idx{0} K\right)$.

    Let us notice that in the presence of the erasing~\eqref{red:rvar} and
    duplicating~\eqref{red:app} rewriting rules, not all reduction sequences
    starting with the same term have to be of equal length. Like in the classic
    \lcalculus, depending on the considered term, some rewriting strategies might
    be more efficient then others.
\end{example}

$`l`y$ enjoys a series of pleasing properties. Most notably, $`l`y$ is
confluent, correctly implements $`b$\nobreakdash-reduction of the classic
\lcalculus, and preserves strong normalisation of closed
terms~\cite{benaissa_briaud_lescanne_rouyer-degli_1996}.  Moreover, the $`y$
fragment, i.e.~$`l`y$ without the~\eqref{red:beta} rule, is terminating. In
other words, each \luterm is $`y$\nobreakdash-normalising as can be shown using,
for instance, polynomial interpretations~\cite{10.1007/3-540-55602-8_161}.  In
the current paper we focus on the normal-order (i.e.~leftmost-outermost)
evaluation strategy of $`y$\nobreakdash-reduction.  For convenience, we assume
the following notational conventions. We use lowercase letters $a,b,c,\ldots$ to
denote arbitrary terms and $s$ (with or without subscripts) to denote
substitutions. Moreover, we write $a \downarrow_n$ to denote the fact that $a$
normalises to its $`y$\nobreakdash-normal form in $n$ normal-order
$`y$\nobreakdash-reduction steps. Sometimes, for further convenience, we also
simply state that $t$ normalises in $n$ steps, without specifying the assumed
evaluation strategy nor the specific rewriting steps and normal form.

\subsection{Regular tree languages}
We base our main construction in the framework of regular tree languages.  In
what follows, we outline their basic characteristics and that of corresponding
regular tree grammars, introducing necessary terminology. We refer the curious
reader to~\cite[Chapter II]{tata2007} for a more detailed exposition.

\begin{definition}[Ranked alphabet]
A \emph{ranked alphabet} $\mathscr{F}$ is a finite set of \emph{function
symbols} endowed with a corresponding \emph{arity} function $\arity \colon
\mathscr{F} \to \nats$. We use $\mathscr{F}_n$ to denote the set of function
symbols of arity $n$, i.e.~function symbols $f \in \mathscr{F}$ such that
$\arity(f) = n$. Function symbols of arity zero are called \emph{constants}.
As a notational convention, we use lowercase letters $f,g,h,\ldots$ to denote
    arbitrary function symbols.
\end{definition}

\begin{definition}[Terms]
    Let $X$ be a finite set of \emph{variables}. Then, the set
    $\mathscr{T}_{\mathscr{F}}(X)$ of \emph{terms over $\mathscr{F}$} is defined
    inductively as follows:
\begin{enumerate}
    \item $X, \mathscr{F}_0 \subset \mathscr{T}_{\mathscr{F}}(X)$;
    \item If $f \in \mathscr{F}_n$ and $\alpha_1,\ldots,\alpha_n \in
        \mathscr{T}_{\mathscr{F}}(X)$, then $f(\alpha_1,\ldots,\alpha_n) \in
        \mathscr{T}_{\mathscr{F}}(X)$.
\end{enumerate}
Terms not containing variables, in other words~elements of
    $\mathscr{T}_{\mathscr{F}}(\emptyset)$, are called \emph{ground terms}.

    As a notational convention, we use lowercase Greek letters
    $\alpha,\beta,\gamma,\ldots$ to denote arbitrary terms.  Whenever it is
    clear from the context, we use the word \emph{term} both to refer to the
    above structures as as well as to denote \luterms.
\end{definition}

\begin{definition}[Regular tree grammars]
    A \emph{regular tree grammar} $\mathscr{G} = (S,\mathscr{N},\mathscr{F},\mathscr{P})$
    is a tuple consisting of:
    \begin{enumerate}
        \item an \emph{axiom} $S \in \mathscr{N}$;
        \item a finite set $\mathscr{N}$ of \emph{non-terminal symbols};
        \item a ranked alphabet $\mathscr{F}$ of \emph{terminal symbols} such
            that $\mathscr{F} \cap \mathscr{N} = \emptyset$; and
        \item a finite set $\mathscr{P}$ of \emph{productions} in form of $N \to \alpha$ such
            that $N \in \mathscr{N}$ and $\alpha \in
            \mathscr{T}_{\mathscr{F}}(\mathscr{N})$.
    \end{enumerate}

A production $N \to \alpha$ is \emph{self-referencing} if $N$ occurs in
    $\alpha$. Otherwise, if $N$ does not occur in $\alpha$, we say that $n \to
    \alpha$ is \emph{regular}. As a notational convention, we use capital
    letters $X,Y,Z\ldots$ to denote arbitrary non-terminal symbols.
\end{definition}

\begin{definition}[Derivation relation]
    The \emph{derivation relation} $\to_\mathscr{G}$ associated with the grammar
    $\mathscr{G} = (S,\mathscr{N},\mathscr{F},\mathscr{P})$ is a relation on
    pairs of terms in $\mathscr{T}_{\mathscr{F}}(\mathscr{N})$ satisfying
    $\alpha \to_\mathscr{G} \beta$ iff there exists a production $N \to \gamma$
    in $\mathscr{P}$ such that after substituting $\gamma$ for some occurrence
    of $N$ in $\alpha$ we obtain $\beta$.  Following standard notational
    conventions, we use $\xrightarrow[]{*}_\mathscr{G}$ to denote the transitive closure
    of $\to_\mathscr{G}$. Moreover, if $\mathscr{G}$ is clear from the context,
    we omit it in the subscript of the derivation relations and simply write
    $\to$ and $\xrightarrow[]{*}$.

    A regular tree grammar $\mathscr{G}$ with axiom $S$ is said to be
    \emph{unambiguous} iff for each ground term $\alpha \in
    \mathscr{T}_{\mathscr{F}}(\emptyset)$ there exist at most one
    \emph{derivation sequence} in form of $S \to \gamma_1 \to \cdots \to
    \gamma_n = \alpha$. Likewise, $N$ is said to be
    \emph{unambiguous in $\mathscr{G}$} iff for each ground term $\alpha \in
    \mathscr{T}_{\mathscr{F}}(\emptyset)$ there exist at most one derivation
    sequence in form of $N \to \gamma_1 \to \cdots \to \gamma_n = \alpha$.
\end{definition}

\begin{definition}[Regular tree languages]
    The \emph{language $L(\mathscr{G})$ generated by $\mathscr{G}$} is the set
    of all ground terms $\alpha$ such that $S \xrightarrow[]{*} \alpha$ where $S$ is the
    axiom of $\mathscr{G}$.  Similarly, the \emph{language generated by term
    $\alpha \in \mathscr{T}_{\mathscr{F}}(\mathscr{N})$ in $\mathscr{G}$},
    denoted as $L_\mathscr{G}(\alpha)$, is the set of of all ground terms
    $\beta$ such that $\alpha \xrightarrow[]{*} \beta$.  Finally, a set $\mathscr{L}$ of
    ground terms over a ranked alphabet $\mathscr{F}$ is said to be a
    \emph{regular tree language} if there exists a regular tree grammar
    $\mathscr{G}$ such that $L(\mathscr{G}) = \mathscr{L}$.
\end{definition}

\begin{example}\label{eg:luterms:grammar}
    The set of \luterms is an example of a regular tree language. The
    corresponding regular tree grammar $\Lambda = (T,
    \mathscr{N}, \mathscr{F}, \mathscr{P})$ consists of
    \begin{itemize}
    \item a set $\mathscr{N}$ of three non-terminal symbols $T$, $S$, $N$
         intended to stand for \luterms, substitutions, and de~Bruijn indices,
            respectively, with $T$ being the axiom of $\Lambda$;
        \item a set $\mathscr{F}$ of terminal symbols, comprised of all
            the symbols of the \lucalculus language, i.e.~term application and
            abstraction, closure $\cdot[\cdot]$, slash $\cdot/$, lift
            $\Uparrow(\cdot)$ and shift $\uparrow$ operators, and the
            successor $\succ(\cdot)$ with the constant $\idx{0}$; and
        \item a set $\mathscr{P}$ of productions
            \begin{align}\label{eq:luterms:tree:grammar}
                \begin{split}
                    T &\to N ~|~ `l T ~|~ T T ~|~ T [S]\\
                    S &\to T/ ~|~ \Uparrow(S) ~|~ \uparrow\\
                    N &\to \idx{0} ~|~ \succ N.
                \end{split}
            \end{align}
    \end{itemize}

    Let us notice that~\eqref{eq:luterms:tree:grammar} consists of five
    self-referencing productions, three for $T$ and one for each $S$ and $N$.
    Moreover, $L(N) \subset L(T)$ as $\mathscr{P}$ includes a production $T \to
    N$.
\end{example}

\section{Reduction grammars}\label{sec:reduction:grammars}

We conduct our construction of $\seq{\mathscr{G}_n}_n$ in an inductive,
incremental fashion. Starting with $\mathscr{G}_0$ corresponding to the set of
\emph{pure terms} (i.e.~\luterms without closures) we
build the $(n+1)$st grammar $\mathscr{G}_{n+1}$ based on the structure of the
$n$th grammar $\mathscr{G}_n$.  First-order rewriting rules of \lucalculus
guarantee a close structural resemblance of both their left- and right-hand
sides, see~\autoref{fig:lambda:ups:rewriting:system}.  Consequently, with
$\mathscr{G}_n$ at hand, we can analyse the right-hand sides of $`y$ rewriting
rules and match them with productions of $\mathscr{G}_n$. Based on their
structure, we then determine the structure
of productions of $\mathscr{G}_{n+1}$ which correspond to respective left-hand
sides. Although such a general idea of constructing $\mathscr{G}_{n+1}$ out of
$\mathscr{G}_n$ is quite straightforward, its implementation requires some
careful amount of detail.

For that reason, we make the following initial preparations. Each grammar
$\mathscr{G}_n$ uses the same, global, ranked alphabet $\mathscr{F}$
corresponding to \lucalculus, see~\cref{eg:luterms:grammar}. The standard
non-terminal symbols $T, S$, and $N$, together with their respective
productions~\eqref{eq:luterms:tree:grammar}, are \emph{pre-defined} in each
grammar $\mathscr{G}_n$.  In addition, $\mathscr{G}_n$ includes $n+1$
non-terminal symbols $G_0,\ldots,G_n$ (the final one being the axiom of
$\mathscr{G}_n$) with the intended meaning that for all $0 \leq k \leq n$, the
language $L_{\mathscr{G}_k}(G_n)$ is equal to the set of terms
$`y$\nobreakdash-normalising in $k$ normal-order steps. In this manner,
building $\seq{\mathscr{G}_n}_n$ amounts to a careful, incremental extension
process, starting with the initial grammar $\mathscr{G}_0$ comprised of the
following extra, i.e.~not included in~\eqref{eq:luterms:tree:grammar},
productions:
\begin{align}
    \begin{split}
    G_0 &\to N~|~`l G_0~|~G_0 G_0\\
    N &\to \succ N~|~\idx{0}.
    \end{split}
\end{align}

In order to formalise the above preparations and the overall presentation of our
construction, we introduce the following, abstract notions of
$`y$\nobreakdash-reduction grammar and later also their simple variants.

\begin{definition}[$`y$-reduction grammars]
    Let $\Lambda = (T, \mathscr{N}, \mathscr{F}, \mathscr{P})$ be the regular
    tree grammar corresponding to \luterms, see~\cref{eg:luterms:grammar}.
    Then, the regular tree grammar
    \begin{equation}
        \mathscr{G}_n = (G_n, \mathscr{N}_n, \mathscr{F}, \mathscr{P}_n)
    \end{equation}
        with
        \begin{itemize}
        \item $\mathscr{N}_n = \mathscr{N} \cup \set{G_0, G_1,\ldots,G_n}$; and
        \item $\mathscr{P}_n$ being a set of productions such that
            $\mathscr{P} \subset \mathscr{P}_n$
        \end{itemize}
        is said to be a \emph{$`y$\nobreakdash-reduction grammar}, \urg in
        short, if for all $0 \leq k \leq n$, the non-terminal $G_k$ is unambiguous in
        $\mathscr{G}_n$, and moreover $L(G_k)$ is equal to the set of \luterms
        $`y$\nobreakdash-normalising in $k$ normal-order $`y$\nobreakdash-steps.
\end{definition}

\begin{definition}[Partial order of sorts]
    The \emph{partial order of sorts} $(\mathscr{N}_n, \preceq)$ is a partial
    order (i.e.~reflexive, transitive, and anti-symmetric relation) on
    non-terminal symbols $\mathscr{N}_n$ satisfying $X \preceq Y$ if and only if
    $L_{\mathscr{G}_n}(X) \subseteq L_{\mathscr{G}_n}(Y)$. For convenience, we
    write $X \varcurlywedge Y$ to denote the greatest lower bound of
    $\set{X,Y}$.~\cref{fig:partial:order:diagram} depicts the partial order
    $(\mathscr{N}_n,\preceq)$.
\begin{figure}[th!]
\centering
\begin{tikzpicture}[every level/.style={sibling distance=20mm/#1}]
    \node [circle,draw] (t){$T$}
        child {node [circle,draw] {$G_0$}
            child {node [circle,draw] {$N$}}}
        child {node [circle,draw] {$G_1$}}
        child {node [circle,draw] {$G_2$}}
        child {node {$\cdots$}}
        child {node [circle,draw] {$G_n$}};
\end{tikzpicture}
\caption{Hasse diagram of the partial order $(\mathscr{N}_n, \preceq)$.}
    \label{fig:partial:order:diagram}
\end{figure}
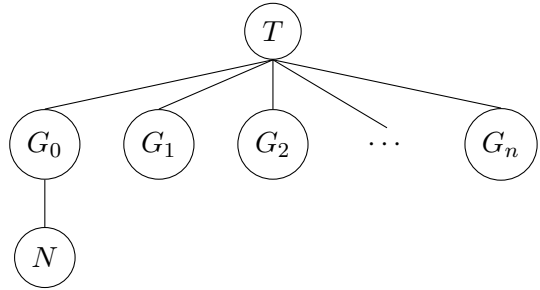
\end{definition}

\begin{remark}\label{rem:sort:poset}
    Let us notice that given the interpretation of $L(G_0),\ldots,L(G_n)$, the
    partial order of sorts $(\mathscr{N}_n, \preceq)$ captures all the
    inclusions among the non-terminal languages within $\mathscr{G}_n$. However,
    in addition, if $X$ and $Y$ are not comparable through $\preceq$, then $L(X)
    \cap L(Y) = \emptyset$ as each term $`y$\nobreakdash-normalises in a unique,
    determined number of steps.
\end{remark}

\begin{definition}[Simple $`y$-reduction grammars]
    A \urg $\mathscr{G}_n$ is said to be \emph{simple} if all its
    self-referencing productions are either productions of the regular tree
    grammar corresponding to \luterms, see~\cref{eg:luterms:grammar}, or are of
    the form
    \begin{equation}
    G_k \to `l G_k~|~G_0 G_k~|~G_k G_0
    \end{equation}
    and moreover, for all regular productions in form of $G_{k} \to \alpha$ in
    $\mathscr{G}_n$, it holds $\alpha \in \mathscr{T}_{\mathscr{F}}(\mathscr{N}
    \cup \set{G_0,\ldots,G_{k-1}})$, i.e.~$\alpha$ does not reference
    non-terminals other than $G_0,\ldots,G_{k-1}$.
\end{definition}

\begin{remark}
Let us note that, in general, $`y$-reduction grammars do not have to be simple.
    Due to the erasing~\eqref{red:rvar} rewriting rule, it is possible to
    construct, \emph{inter alia}, more involved self-referencing productions.
    Nonetheless, for technical convenience, we will maintain the simplicity of
    constructed grammars $\seq{\mathscr{G}_n}_n$.

    Also, let us remark that the above definition of simple
    $`y$\nobreakdash-reduction grammars asserts that if $\mathscr{G}_{n+1}$ is a
    simple \urg, then, by a suitable truncation, it is possible to obtain all
    the $`y$\nobreakdash-reduction grammars $\mathscr{G}_0$ up to
    $\mathscr{G}_n$.  Consequently, $\mathscr{G}_{n+1}$ contains, in a proper
    sense, all the proceeding grammars $\mathscr{G}_0,\ldots,\mathscr{G}_n$.
\end{remark}

\section{Finite intersection
partitions}\label{sec:finite:intersection:partition}

The main technical ingredient in our construction of $\seq{\mathscr{G}_n}_n$ are
finite intersection partitions.

\begin{definition}[Finite intersection partition]
    Let $\alpha, \beta$ be two terms in $\mathscr{T}_{\mathscr{F}}(X)$.  Then,
    a \emph{finite intersection partition}, \fip in short, of $\alpha$ and
    $\beta$ is a finite set $\Pi(\alpha,\beta) = \set{\pi_1,\ldots,\pi_n}
    \subset \mathscr{T}_{\mathscr{F}}(X)$ such that $L(\pi_i) \cap L(\pi_j) =
    \emptyset$ whenever $i \neq j$, and $\bigcup_i L(\pi_i) = L(\alpha) \cap
    L(\beta)$.
\end{definition}

Let us note that, \emph{a priori}, it is not immediately clear if
$\Pi(\alpha,\beta)$ exists for $\alpha$ and $\beta$ in the term algebra
$\mathscr{T}(\mathscr{N}_n)$ associated with the simple \urg $\mathscr{G}_n$ nor
whether there exists an algorithmic method of its construction. The following
result states that both questions can be settled in the affirmative.

\begin{lemma}[Constructible finite intersection partitions]\label{lem:computable:fip}
Assume that $\mathscr{G}_n$ is a simple $`y$\nobreakdash-reduction grammar.  Let
    $\alpha,\beta$ be two (not necessarily ground) terms in
    $\mathscr{T}(\mathscr{N}_n)$ where $\mathscr{N}_n$ is the set of
    non-terminal symbols of $\mathscr{G}_n$. Then, $\alpha$ and $\beta$ have a
    computable finite intersection partition $\Pi(\alpha,\beta)$.
\end{lemma}

In order to prove~\cref{lem:computable:fip}, we resort to the following technical
notion of \emph{term potential} used to embed terms into the well-founded set of
natural numbers.

\begin{definition}[Term potential]\label{def:term:potential}
    Let $\alpha \in \mathscr{T}(\mathscr{N}_n)$ be a term and $\mathscr{G}_n$ be
    a simple \urg. Let ${\sf Prod}_{\mathscr{G}_n}(X)$ denote the set consisting of right-hand
    sides of regular productions in form of $X \to \beta$ in $\mathscr{G}_n$.
    Then, the \emph{potential} $\pi(\alpha)$ of $\alpha$ in $\mathscr{G}_n$ is
    defined inductively as follows:
    \begin{itemize}
        \item If $\alpha = f(\alpha_1,\ldots,\alpha_m)$, then
            $\pi(\alpha) = 1 + \displaystyle\sum_{i=1}^{m} \pi(\alpha_i)$;
        \item If $\alpha = X \in \set{T,S,N}$, then
            $\pi(\alpha) = 1 + \displaystyle \max \set{\pi(\gamma)~|~\gamma \in
            {\sf Prod}_{\mathscr{G}_n}(X)}$;
        \item If $\alpha = G_k$ for some $0 \leq k \leq n$, then $\pi(\alpha) = 1 + \displaystyle \max
            \set{\pi(\gamma)~|~ \gamma \in \bigcup_{i=0}^k
                {\sf Prod}_{\mathscr{G}_n}(G_i)}$.
    \end{itemize}

    Let us note that $\pi$ is well-defined as, by assumption, ${\sf
    Prod}_{\mathscr{G}_n}(X) \neq \emptyset$
    for all simple $`y$\nobreakdash-reduction grammars $\mathscr{G}_n$;
    otherwise $L(G_k)$ could not span the whole set of \luterms
    $`y$\nobreakdash-normalising in $k$ normal-order steps.
    Moreover, $\pi$ has the following crucial properties:
    \begin{itemize}
    \item For each term $\alpha$ we have $\pi(\alpha) \geq 1$;
    \item If $\alpha$ is a proper subterm of $\beta$, then $\pi(\alpha) <
        \pi(\beta)$;
    \item If $X \to \alpha$ is a regular production,
        then $\pi(\alpha) < \pi(X)$; and
    \item For each $G_i$ and $G_j$, it holds $\pi(G_i) \leq \pi(G_j)$ whenever $i
        \leq j$.
    \end{itemize}
\end{definition}

\begin{example}
Note that the term potential of $N$ associated with de~Bruijn indices is equal
    to ${\pi(N) = 2}$ as $\pi(\idx{0}) = 1$. Since $T \to N$ is the single
    regular production starting with $T$ on its left-hand side, the potential
    $\pi(T)$ is therefore equal to $3$. Consequently, we also have $\pi(S) = 5$
    as witnessed by the regular production $S \to T /$. Finally, since $\pi(N) =
    2$ it holds $\pi(G_0) = 3$ and so, for instance, we also have $\pi(G_0 G_0)
    = 7$.
\end{example}

Productions of a simple $\mathscr{G}_n$ cannot reference non-terminals other
than $G_0,\ldots,G_n$. Since the potential of $G_{k+1}$ is defined in terms of the
potential of its regular productions, this means that $\pi(G_{k+1})$ depends, in an
implicit manner, on the potentials $\pi(G_0),\ldots,\pi(G_{k})$. Note that
this constitutes a traditional inductive definition. In order to compute the
potential of a given term $\alpha$, we start with computing the potential of
associated non-terminals. In particular, we find the values
$\pi(G_0),\ldots,\pi(G_n)$ in ascending order. Afterwards, we can recursively
decompose $\alpha$ and calculate its potential based on the potential of
non-terminal symbols occurring in $\alpha$. Note that the same scheme holds, in
particular, for the right-hand sides of self-referencing productions.

\begin{definition}[Conservative productions]
    A self-referencing production $X \to f(\alpha_1,\ldots,\alpha_n)$ is said to be
    \emph{conservative} if $\pi(\alpha_i) \leq \pi(X)$ for all $1 \leq i \leq
    n$.
\end{definition}

\begin{remark}\label{rem:conservative:productions}
    Conservative productions play a central role in the algorithmic cosntruction
    of finite intersection partitions $\Pi(\alpha,\beta)$. In particular, let us
    remark that all self-referencing productions of a simple \urg
    $\mathscr{G}_n$, as listed in~\cref{fig:self:referencing:prods}, are at the
    same time conservative.
\begin{figure}[hbt!]
\centering
\begin{subfigure}{.45\textwidth}
    \begin{align*}\label{eq:luterms:spec}
    \begin{split}
        G_k &\to `l G_k~|~G_0 G_k~|~G_k G_0\\
        T &\to `l T~|~T T~|~T[S]\\
    \end{split}
\end{align*}
\end{subfigure}%
\begin{subfigure}{.45\textwidth}
    \begin{align*}
    \begin{split}
        S &\to~\Uparrow(S)\\
        N &\to \succ N
    \end{split}
    \end{align*}
\end{subfigure}
    \caption{Self-referencing productions of simple $`y$\nobreakdash-reduction
    grammars.}\label{fig:self:referencing:prods}
\end{figure}
\end{remark}

With the technical notions of term potential and conservative productions, we
are ready to present a recursive procedure $\texttt{fip}_k$ constructing
$\Pi(\alpha,\beta)$ for arbitrary terms $\alpha,\beta$ within the scope of a
simple \urg $\mathscr{G}_k$.~\cref{fig:fip:code} provides the functional
pseudocode describing $\texttt{fip}_k$.

\begin{figure}[ht]
\begin{Verbatim}[commandchars=\\\{\},codes={\catcode`\$=3\catcode`\^=7\catcode`\_=8}]
    \PYG{k}{fun} \PYG{n}{fip}\PYG{esc}{$_{k}$} \PYG{esc}{$\alpha$} \PYG{esc}{$\beta$} \PYG{o}{:=}
        \PYG{k}{match} \PYG{esc}{$\alpha$}\PYG{o}{,} \PYG{esc}{$\beta$} \PYG{k}{with}
        \PYG{o}{|} \PYG{esc}{$f(\alpha_1,\ldots,\alpha_n)$}\PYG{o}{,} \PYG{esc}{$g(\beta_1,\ldots,\beta_m)$} \PYG{esc}{$\Rightarrow$}
            \PYG{k}{if} \PYG{esc}{$f \neq g \lor n \neq m$} \PYG{k}{then} \PYG{esc}{$\emptyset$} \PYG{c}{(*}\PYG{c}{ symbol clash }\PYG{c}{*)} \PYG{esc}{$\label{code:line:4}$}
            \PYG{k}{else} \PYG{k}{if} \PYG{esc}{$n = m = 0$} \PYG{k}{then} \PYG{esc}{$\set{f}$}\PYG{esc}{$\label{code:line:5}$}
            \PYG{k}{else} \PYG{k}{let} \PYG{esc}{$\Pi_i$} \PYG{o}{:=} \PYG{n}{fip}\PYG{esc}{$_{k}$} \PYG{esc}{$\alpha_i$} \PYG{esc}{$\beta_i$}\PYG{o}{,} \PYG{esc}{$\forall$} \PYG{esc}{$1 \leq i \leq n$}\PYG{esc}{$\label{code:line:6}$}
                    \PYG{k}{in} \PYG{esc}{$\set{f(\pi_1,\ldots,\pi_n)\, |\, (\pi_1,\ldots,\pi_n) \in \Pi_1 \times \cdots \times \Pi_n}$}\PYG{esc}{$\label{code:line:7}$}

        \PYG{o}{|} \PYG{esc}{$X$}\PYG{o}{,} \PYG{esc}{$Y$} \PYG{esc}{$\Rightarrow$}\PYG{esc}{$\label{code:line:9}$}
            \PYG{esc}{$\set{X \varcurlywedge Y\, |\, X \preceq Y \lor Y \preceq X}$}\PYG{esc}{$\label{code:line:10}$}

        \PYG{o}{|} \PYG{esc}{$f(\alpha_1,\ldots,\alpha_n)$}\PYG{o}{,} \PYG{esc}{$X$} \PYG{esc}{$\Rightarrow$}
            \PYG{n}{fip}\PYG{esc}{$_{k}$} \PYG{esc}{$X$} \PYG{esc}{$f(\alpha_1,\ldots,\alpha_n)$} \PYG{c}{(*}\PYG{c}{ flip arguments }\PYG{c}{*)}\PYG{esc}{$\label{code:line:13}$}

        \PYG{o}{|} \PYG{esc}{$X$}\PYG{o}{,} \PYG{esc}{$f(\alpha_1,\ldots,\alpha_n)$} \PYG{esc}{$\Rightarrow$}
            \PYG{k}{let} \PYG{esc}{$\Pi_{\gamma}$} \PYG{o}{:=} \PYG{n}{fip}\PYG{esc}{$_{k}$} \PYG{esc}{$\gamma$} \PYG{esc}{$f(\alpha_1,\ldots,\alpha_n)$}\PYG{o}{,} \PYG{esc}{$\forall$} \PYG{esc}{$(X \to \gamma) \in \mathscr{G}_k$}\PYG{esc}{$\label{code:line:16}$}
                \PYG{k}{in} \PYG{esc}{$\bigcup_{(X \to \gamma) \in \mathscr{G}_k} \Pi_{\gamma}$}\PYG{esc}{$\label{code:line:17}$}
    \PYG{k}{end}\PYG{o}{.}
\end{Verbatim}



    \caption{Pseudocode of the $\texttt{fip}_k$ procedure computing
    $\Pi(\alpha,\beta)$.}\label{fig:fip:code}
\end{figure}

\begin{proof}{(of~\cref{lem:computable:fip})}
Induction over the total potential of $\alpha$ and $\beta$.

Let us start with the base case $\pi(\alpha) + \pi(\beta) = 2$. Note that both
    $\alpha$ and $\beta$ have to be constant ground terms (the potential of
    non-terminals in $\mathscr{N}_k$ is at least $2$). If $\alpha \neq \beta$,
    then certainly $L(\alpha) \cap L(\beta) = \emptyset$ and so
    $\Pi(\alpha,\beta) = \emptyset$, see~\lref{code:line:4}. Otherwise if
    $\alpha = \beta$, then both $L(\alpha) \cap L(\beta) = \Pi(\alpha,\beta) =
    \set{\alpha} = \set{\beta}$; hence, $\texttt{fip}_k$ returns a correct
    intersection partition $\Pi(\alpha, \beta)$, see~\lref{code:line:5}.  And
    so, assume that $\pi(\alpha) + \pi(\beta) > 2$.  Depending on the joint
    structure of both $\alpha$ and $\beta$ we have to consider three cases.

    \paragraph{\textbf{Case 1.}}
    Suppose that $\alpha = f(\alpha_1,\ldots,\alpha_n)$ and $\beta=
    g(\beta_1,\ldots,\beta_m)$ for some function symbols $f$ and $g$ with
    $n,m \geq 0$.  Certainly, if either $f \neq g$ or $n \neq m$, then
    $L(\alpha) \cap L(\beta) = \emptyset$. In consequence, $\emptyset$ is
    the sole valid \fip of both $\alpha$ and $\beta$, see~\lref{code:line:4}.

    So, let us assume that both $f = g$ and $n = m$. Moreover, we can also
    assume that $n,m \geq 1$ as the trivial case $n = m = 0$ cannot occur under
    the working assumption $\pi(\alpha) + \pi(\beta) > 2$. Take an arbitrary
    $\delta \in L(\alpha) \cap L(\beta)$.  Note that $\delta$ takes the form of
    $\delta = f(\delta_1,\ldots,\delta_n)$ for some ground terms
    $\delta_1,\ldots,\delta_n$. By induction, for all $1 \leq i \leq n$, the
    recursive call $\texttt{fip}_k\, \alpha_i\, \beta_i$ yields a finite
    intersection partition $\Pi(\alpha_i,\beta_i)$ of $\alpha_i$ and $\beta_i$.
    Since $\delta_i \in L(\alpha_i) \cap L(\beta_i)$ there exists a unique
    $\pi_i \in \Pi(\alpha_i,\beta_i)$ such that $\delta_i \in L(\pi_i)$.
    Accordingly, for $\delta = f(\delta_1,\ldots,\delta_n)$ there exists a
    unique term $\pi = f(\pi_1,\ldots,\pi_n)$ in $\texttt{fip}_k\, \alpha\,
    \beta$ such that $\delta \in L(\pi)$.

    Conversely, take an arbitrary ground term $\delta =
    f(\delta_1,\ldots,\delta_n) \in L(\pi)$ for some $\pi \in \texttt{fip}_k\,
    \alpha\, \beta$. Note that $\pi$ takes the form $\pi =
    f(\pi_1,\ldots,\pi_n)$, see~\lref{code:line:7}. Since $\delta \in L(\pi)$,
    we know that $\delta_i \in L(\pi_i)$, for each $1\leq i \leq n$.  Moreover,
    $\pi_i \in \Pi_i$ by the construction of $\texttt{fip}_k\, \alpha\, \beta$,
    see~\lref{code:line:6}.  Following the inductive hypothesis that $\Pi_i$ is
    a \fip of $\alpha_i$ and $\beta_i$, we notice that $L(\pi_i) \subseteq
    L(\alpha_i) \cap L(\beta_i)$.  Consequently, $\delta_i \in L(\alpha_i) \cap
    L(\beta_i)$ and so $\delta \in L(\alpha) \cap L(\beta)$.

    \paragraph{\textbf{Case 2.}}
    Suppose that $\alpha = X$ and $\beta = Y$ are two non-terminal symbols in
    $\mathscr{N}_k$, see~\lref{code:line:9}. Let us consider the sort poset
    $(\mathscr{N}_k,\preceq)$ associated with $\mathscr{G}_k$.  Assume that $X$
    and $Y$ are comparable through $\preceq$ (w.l.o.g.~let $X \preceq Y$).
    Consequently, $L(X) \subseteq L(Y)$ and so $X \varcurlywedge Y = X$. Clearly,
    $\set{X}$ is a valid \fip, see~\lref{code:line:10}. On the other hand, if
    $X$ and $Y$ are incomparable in the sort poset associated with
    $\mathscr{G}_k$, it means that $L(X) \cap L(Y) = \emptyset$ and so
    $\Pi(\alpha,\beta) = \emptyset$, see~\cref{rem:sort:poset}.

    \paragraph{\textbf{Case 3.}}
    Suppose w.l.o.g~that $\alpha = X$ and $\beta$ takes the
    form $\beta = f(\beta_1,\ldots,\beta_n)$ with $n \geq 0$. Note that
    $\texttt{fip}_k$ flips its arguments if necessary, see~\lref{code:line:13}.
    Take an arbitrary $\delta \in L(\alpha) \cap L(\beta)$.  Note that
    from the form of $\beta$ we know that $\delta = f(\delta_1,\ldots,\delta_n)$
    for some ground terms $\delta_1,\ldots,\delta_n$ ($n \geq 0$). Since $\alpha = X$
    is a non-terminal symbol which, by assumption, is unambiguous in
    $\mathscr{G}_k$, there exists a unique production $X \to \gamma$
    such that $\delta \in L(\gamma)$.

    If $X \to \gamma$ is regular (i.e.~$X$ does not occur in $\gamma$),
    then $\pi(\gamma) < \pi(X)$ and so $\pi(\gamma) + \pi(\beta) <
    \pi(X) + \pi(\beta)$. Hence, by induction, $\texttt{fip}_k\, \gamma\,
    \beta$ constructs a finite intersection partition
    $\Pi(\gamma,\beta)$ with a unique $\pi \in
    \Pi(\gamma,\beta) \subset \Pi(X,\beta)$ such that $\delta \in
    L(\pi)$, see~\lref{code:line:17}.

    Let us therefore assume that $X \to \gamma$ is not regular, but instead
    self-referencing (i.e.~$X$ occurs in $\gamma$). In such a case $\pi(X) \leq
    \pi(\gamma)$ and so we cannot directly apply the induction hypothesis to
    $\texttt{fip}_k\, \gamma\, \beta$.  Note however, that since $\delta \in
    L(\gamma) \cap L(\beta)$ and $\gamma \neq X$, the term $\gamma$ must be of
    form $\gamma = f(\gamma_1,\ldots,\gamma_n)$ as otherwise $\delta \not\in
    L(\gamma)$. Furthermore if $n = 0$, then trivially $\gamma = \beta = f$,
    see~\lref{code:line:5}. Hence, let us assume that $n \geq 1$.
    It follows that $\texttt{fip}_k$ proceeds to construct finite
    intersection partitions for respective pairs of arguments $\gamma_i$ and
    $\beta_i$. However, since $X \to f(\gamma_1,\ldots,\gamma_n)$ is
    conservative (see~\cref{rem:conservative:productions}), it holds
    $\pi(\gamma_i) \leq \pi(X)$. At the same time, $\pi(\beta_i) < \pi(\beta)$;
    hence, by induction we can argue that $\texttt{fip}_k\, \gamma_i\, \beta_i$
    constructs a proper intersection partition $\Pi(\gamma_i,\beta_i)$ for each
    $1 \leq i \leq n$, see~\lref{code:line:7}.  There exists therefore a unique
    term $\pi = f(\pi_1,\ldots,\pi_n) \in \texttt{fip}_k\, X\, \beta$ such that
    $\delta \in L(\pi)$, see~\lref{code:line:7} and~\lref{code:line:17}.

    Conversely, take an arbitrary ground term $\delta =
    f(\delta_1,\ldots,\delta_n) \in L(\pi)$ for some $\pi \in \texttt{fip}_k\,
    X\, \beta$ $(n \geq 0)$. By definition, $\texttt{fip}_k$ proceeds to invoke itself on
    pairs of arguments $\gamma$ and $\beta$ where $\gamma$ is the right-hand
    side of a production $X \to \gamma$ in $\mathscr{G}_k$,
    see~\lref{code:line:16}, and returns the set-theoretic union of recursively
    obtained outcomes. There exists therefore some $\gamma$ such that $\pi \in
    \texttt{fip}_k\, \gamma\, \beta$. If $X \to \gamma$ is regular, then by
    induction, $\texttt{fip}_k\, \gamma\, \beta$ constructs a \fip for both
    $\gamma$ and $\beta$. Consequently, it holds $\delta \in L(\pi) \subset
    L(\gamma) \cap L(\beta) \subset L(X) \cap L(\beta)$.  Assume therefore that
    $X \to \gamma$ is not regular, but instead self-referencing. As before, we
    cannot directly argue about $\texttt{fip}_k\, \gamma\, \beta$ since the
    total potential of $\gamma$ and $\beta$ exceeds the potential of $X$ and
    $\beta$.  However, since $\pi \in \texttt{fip}_k\, \gamma\, \beta$ and
    $\gamma \neq X$, we note that $\gamma$ takes form $\gamma =
    f(\gamma_1,\ldots,\gamma_n)$, see~\lref{code:line:4}. If $f$ is a constant
    symbol, then certainly $\texttt{fip}_k\, \gamma\, \beta$ outputs a proper
    \fip. Otherwise, $\texttt{fip}_k$ proceeds to invoke itself recursively on
    respective pairs of arguments $\gamma_i$ and $\beta_i$.  Since $X \to
    f(\gamma_1,\ldots,\gamma_n)$ is conservative, we know that, by induction,
    $\texttt{fip}_k\, \gamma_i\, \beta_i$ constructs finite intersection
    partitions $\Pi(\gamma_i,\beta_i)$ for all pairs $\gamma_i$ and $\beta_i$.
    Certainly, $\delta_i \in L(\pi_i)$ for some $\pi_i \in
    \Pi(\gamma_i,\beta_i)$; hence $\delta \in L(\pi)$ where $\pi =
    f(\pi_1,\ldots,\pi_n) \in \texttt{fip}_k\, \gamma\, \beta$. It follows that
    $\delta \in L(\pi) \cap L(\beta) \subset L(X) \cap L(\beta)$, which finishes
    the proof.
\end{proof}

\begin{remark}
Note that the termination of $\texttt{fip}_k$ is based on the fact that
    all self-referencing productions of simple $`y$\nobreakdash-reduction
    grammars are at the same time conservative. Indeed, $\texttt{fip}_k$ does
    not terminate in the presence of non-conservative productions. Consider the
    non-conservative production $X \to f(f(X))$. Note that
    \begin{align}
        \begin{split}
            \texttt{fip}_k(f(f(X)), f(X)) &\to
            \texttt{fip}_k(f(X), X)\\ &\to \texttt{fip}_k(X, f(X))\\
            &\to \texttt{fip}_k(f(f(X)), f(X))\\
            &\to \cdots
        \end{split}
    \end{align}
\end{remark}

\begin{remark}\label{rem:T:optimisation}
    It is possible to optimise $\texttt{fip}_k$, as presented
    in~\cref{fig:fip:code}, and (potentially) shrink the
    size of $\Pi(\alpha,\beta)$ by including an additional pattern in form of
\begin{Verbatim}[commandchars=\\\{\},codes={\catcode`\$=3\catcode`\^=7\catcode`\_=8}]
        \PYG{esc}{$T$}\PYG{o}{,} \PYG{esc}{$f(\alpha_1,\ldots,\alpha_n)$} \PYG{esc}{$\Rightarrow$} \PYG{esc}{$\set{f(\alpha_1,\ldots,\alpha_n)}$}
\end{Verbatim}
\end{remark}

\begin{remark}
Our finite intersection partition algorithm resembles a variant of Robinsons's
    unification algorithm~\cite{Robinson:1965:MLB:321250.321253} applied to
    many-sorted term algebras with a tree hierarchy of sorts, as investigated by
    Walther, cf.~\cite{Walther:1988}. It becomes even more apparent once the
    correspondence between sorts, as stated in the language of many-sorted term
    algebra, and the tree-like hierarchy of non-terminal symbols in
    $`y$\nobreakdash-reduction grammars is established,
    see~\cref{fig:partial:order:diagram}.
\end{remark}

\section{The construction of simple $`y$-reduction grammars}\label{sec:the:construction}
Equipped with constructible, finite intersection partitions, we are now ready to
describe the generation procedure for $\seq{\mathscr{G}_n}_n$. We begin with
establishing a convenient \emph{verbosity} invariant maintained during the
construction of $\seq{\mathscr{G}_n}_n$.

\begin{definition}[Closure width]
    Let $\alpha$ be a term in $\mathscr{T}_{\mathscr{F}}(X)$ for some finite set $X$. Then,
    $\alpha$ has \emph{closure width} $w$, if $w$ is the largest non-negative
    integer such that $\alpha$ is of form $\chi[\sigma_1]\cdots[\sigma_w]$ for
    some term $\chi$ and substitutions $\sigma_1,\ldots,\sigma_w$.
    For convenience, we refer to $\chi$ as the \emph{head} of $\alpha$ and to
    $\sigma_1,\ldots,\sigma_w$ as its \emph{tail}.
\end{definition}

\begin{definition}[Verbose $`y$-reduction grammars]
    A \urg $\mathscr{G}_n$ is said to be \emph{verbose} if none of its
    productions takes the form $X \to G_k [\sigma_1] \cdots [\sigma_w]$ for
    some arbitrary non-negative $w$ and $k$.
\end{definition}

Simple, verbose $`y$\nobreakdash-reduction grammars admit a neat structural
feature. Specifically, their productions preserve closure width of
generated terms.

\begin{lemma}\label{lem:verbose:urg}
    Let $\mathscr{G}_n$ be a simple, verbose \urg. Then, for each production
    $G_n \to \chi [\sigma_1] \cdots [\sigma_w]$ in $\mathscr{G}_n$ such that its
    right-hand side is of closure width $w$, either $\chi = N$ or $\chi$ is in
    form $\chi = f(\alpha_1,\ldots,\alpha_m)$ for some non-closure function symbol $f$ of
    arity $m$.
\end{lemma}

\begin{proof}
    Suppose that neither $\chi = N$ nor $\chi = f(\alpha_1,\ldots,\alpha_m)$.
    Since $\mathscr{G}_n$ is simple, it follows that either $\chi = T$ or $\chi
    = G_k$ for some $0 \leq k \leq n$. However, due to the verbosity of
    $\mathscr{G}_n$ we know that $\chi \neq G_k$ and
    so it must hold $\chi = T$. Consider the following
    inductive family of terms:
    \begin{equation}
            \delta_1 = \idx{0}[\Uparrow(\uparrow)]
            \qquad \text{whereas} \qquad \delta_{n+1} = \idx{0}[\delta_n/].
    \end{equation}
    By construction, we note that $\delta_n \downarrow_n$.  Let $s_1,\ldots,s_w$ be substitutions
    satisfying $s_i \in L(\sigma_i)$. Note that $\delta :=
    \delta_{n+1}[s_1]\cdots[s_w] \in L(T[\sigma_1]\cdots[\sigma_w])$;
    hence, simultaneously $\delta$ reduces in $n$ steps, as $\delta \in
    L(G_n)$, and in at least $n+1$ steps, contradiction.
\end{proof}

\begin{lemma}\label{lem:verbose:closure:width}
    Let $\mathscr{G}_n$ be a simple, verbose \urg. Then, for each production
    $G_n \to \chi[\sigma_1]\cdots[\sigma_w]$ in $\mathscr{G}_n$ such that its
    right-hand side is of closure width $w$, and ground term $\delta \in
    L(\chi[\sigma_1]\cdots[\sigma_w])$ it holds that $\delta$ is of closure
    width $w$.
\end{lemma}

\begin{proof}
    Direct consequence of~\cref{lem:verbose:urg}.
\end{proof}

The following $\varphi$\nobreakdash-matchings are the central tool used in the
construction of new reduction grammars. Based on finite intersection partitions,
$\varphi$\nobreakdash-matchings provide a simple \emph{template recognition} mechanism
which allows us to match productions in $\mathscr{G}_n$ with right-hand side of
$`y$ rewriting rules.

\begin{definition}[$\varphi$\nobreakdash-matchings]
    Let $\mathscr{G}_n$ be a simple, verbose \urg and $\varphi = \chi [\tau_1] \cdots
    [\tau_d] \in \mathscr{T}(\mathscr{N}_n)$ be a \emph{template} (term) of
    closure width $d$.  Assume that $X \to \gamma[\sigma_1]\cdots[\sigma_w]$ is
    a production of $\mathscr{G}_n$ which right-hand side has closure width $w$.
    Furthermore, let
    \begin{equation}
        \Delta_{\varphi}(\gamma[\sigma_1]\cdots[\sigma_w]) =
            \Pi(\gamma[\sigma_1]\cdots[\sigma_w], \chi [\tau_1] \cdots [\tau_d]
                \underbrace{[S] \cdots [S]}_{w-d \text{ times}})
    \end{equation}
    be the set of \emph{$\varphi$\nobreakdash-matchings of
    $\gamma[\sigma_1]\cdots[\sigma_w]$}.

    Then, the set $\Delta_{\varphi}^n$ of \emph{$\varphi$\nobreakdash-matchings}
    of $\mathscr{G}_n$ is defined as
    \begin{equation}
        \Delta_\varphi^n = \bigcup \set{\Delta_\varphi(\gamma)~|~G_n \to \gamma \in
        \mathscr{G}_n}.
    \end{equation}
    For further convenience, we write
    $\varphi_{\langle i \rangle}$ to denote the template $\varphi = \chi
    [\tau_1] \cdots [\tau_d]$ with $i$ copies of $[S]$ appended to its original
    tail, i.e.~$\varphi_{\langle i \rangle} = \chi [\tau_1] \cdots [\tau_d] \underbrace{[S]
    \cdots [S]}_{i~\text{times}}$.
\end{definition}

\begin{table}[ht]
    \caption{$`y$-rewriting rules with respective
    templates and production schemes.}\label{tab:templates}
\begin{tabular}{c|c|c|c}
    & Rewriting rule & Template $\varphi$ & Production scheme
    $\Delta_\varphi^n \mapsto (G_{n+1} \to \gamma)$\\
  \hline
    (App)& $(a b)[s] \to a [s] (b [s])$ & $T[S] (T[S])$ &
       $\alpha [\tau_1] (\beta [\tau_2]) [\sigma_1]\cdots[\sigma_w] \mapsto
       (\alpha \beta) [\tau] [\sigma_1]\cdots[\sigma_w]^\dagger$\\
  \hline
    (Lambda)& $(`l a)[s] \to `l(a[\Uparrow (s)])$ & $`l (T[\Uparrow (S)])$
    & $`l (\alpha[\Uparrow (\sigma)]) [\sigma_1] \cdots [\sigma_w] \mapsto (`l \alpha) [\sigma] [\sigma_1] \cdots
    [\sigma_w]$ \\
  \hline
    (FVar)& $\idx{0} [a/] \to a$ & $T$ &see~\cref{rem:production:scheme:closure:width}\\
  \hline
    (RVar)& $(\succ \idx{n}) [a/] \to \idx{n}$ & $N$ &
    $\alpha [\sigma_1]\cdots[\sigma_w] \mapsto (\succ \alpha)[T/][\sigma_1]\cdots[\sigma_w] $\\
  \hline
    (FVarLift)& $\idx{0} [\Uparrow(s)] \to \idx{0}$ & $\idx{0}$ &
    $\idx{0} [\sigma_1]\cdots[\sigma_w] \mapsto \idx{0} [\Uparrow(S)][\sigma_1]\cdots[\sigma_w]$\\
  \hline
    (RVarLift)& $(\succ \idx{n}) [\Uparrow(s)] \to \idx{n}[s][\uparrow]$ &
    $N[S][\uparrow]$ & $\alpha [\sigma] [\uparrow] [\sigma_1]\cdots[\sigma_w]
    \mapsto (\succ \alpha) [\Uparrow(\sigma)][\sigma_1]\cdots[\sigma_w]$\\
  \hline
    (VarShift)& $\idx{n}[\uparrow] \to \succ \idx{n}$ & $\succ N$ &
    $(\succ \alpha) [\sigma_1]\cdots[\sigma_w] \mapsto \alpha
    [\uparrow][\sigma_1]\cdots[\sigma_w]$\\
\end{tabular}
    \begin{flushleft}
        {\footnotesize $^\dagger$For each $\tau \in \Pi(\tau_1,\tau_2)$,
        see~\cref{rem:app:construction}.}
    \end{flushleft}
    \begin{flushleft}
        {\footnotesize The final column contains right-hand sides of respective
        productions.}
    \end{flushleft}
\end{table}

In what follows we use computable intersection partitions in our iterative
construction of $\seq{\mathscr{G}_n}_n$. Recall that if $\mathscr{G}_n$ is
simple then, \emph{inter alia}, self-referencing productions starting with the
non-terminal $G_n$ take the form
\begin{equation}\label{eq:Gn:selfref}
G_n \to `l G_n~|~G_0 G_n~|~G_n G_0.
\end{equation}
If $t \downarrow_n$ (for $n \geq 1$) but it does not strat with a head
$`y$\nobreakdash-redex, then it must be of form $t = `l a$ or $t = a b$. In
the former case, it must hold $a \downarrow_n$; hence the pre-defined production
$G_n \to `l G_n$ in $\mathscr{G}_n$. In the latter case, it must hold
$a \downarrow_k$ whereas $b \downarrow_{n-k}$ for some $0 \leq k \leq n$.
And so, it means that we have to include productions in form of $G_n \to G_k G_{n-k}$
for all $0 \leq k \leq n$ in $\mathscr{G}_n$; in particular, the already
mentioned two self-referencing productions, see~\eqref{eq:Gn:selfref}.

Remaining terms have to start with head redexes. Each of these head
$`y$\nobreakdash-redexes is covered by a dedicated set of productions. The
following~\cref{lem:lambda:rule:productions} demonstrates how
$\varphi$\nobreakdash-matchings and, in particular, finite intersection
partitions can be used for that purpose.

\begin{lemma}\label{lem:lambda:rule:productions}
    Let $\varphi = `l (T[\Uparrow(S)])$ be the template corresponding to
    the~\eqref{red:lambda} rewriting rule, see~\cref{tab:templates}, and $t = (`l
    a)[s][s_1]\cdots[s_w]$. Then, $t \downarrow_{n+1}$ if and only if there exists a unique
    term $\pi = \left(`l (\alpha [\Uparrow (\sigma)])\right)[\sigma_1]\cdots[\sigma_w] \in
    \Delta_\varphi^n$ such that $t \in L((`l \alpha) [\sigma]
    [\sigma_1]\cdots[\sigma_w])$.
\end{lemma}

\begin{proof}
Let $t = (`l a)[s][s_1]\cdots[s_w] \downarrow_{n+1}$ where $w \geq 0$. Since $t$
    admits a head $`y$\nobreakdash-redex, we note that $t \to t' = \left(`l (a
    [\Uparrow(s)])\right)[s_1]\cdots[s_w] \downarrow_n$. By assumption, $\mathscr{G}_n$
    is simple, hence there exists a unique production $G_n \to \gamma$ in
    $\mathscr{G}_n$ such that $t' \in L(\gamma)$.
    Consider the set $\Delta_\varphi^n$. Since $G_n \to \gamma$ is the unique
    production satisfying $t' \in L(\gamma)$, it follows that for each production
    $G_n \to \gamma'$ in $\mathscr{G}_n$ such that $\gamma' \neq \gamma$ and all
    $\pi \in \Delta_\varphi(\gamma')$ it holds $t' \not\in L(\pi)$.  Let us
    therefore focus on the set $\Delta_\varphi(\gamma)$ of
    $\varphi$\nobreakdash-matchings limited to $\gamma$.

    By assumption, $\mathscr{G}_n$ is not only simple but also verbose.
    Consequently, we
    know that $\gamma$ retains the closure width of generated terms,
    see~\cref{lem:verbose:closure:width}. It follows that $\gamma$ has closure
    width $w$ and takes the form $\gamma = \chi[\tau_1]\cdots[\tau_w]$.
    Certainly, $t' \in L(\varphi_{\langle w \rangle})$. Moreover,
    $\Delta_\varphi(\gamma) = \Pi(\gamma,\varphi_{\langle w \rangle})$.
    There exists therefore  a
    unique $\pi \in \Pi(\gamma,\varphi_{\langle w \rangle})$ such that
    $t' \in L(\pi)$.  Given the fact that the head of $\varphi_{\langle w
    \rangle}$ is equal to
    $\varphi = `l(T[\Uparrow(S)])$ we note that $\pi$ must be of form $\pi = `l
    (\alpha[\Uparrow(\sigma)]) [\sigma_1]\cdots[\sigma_w]$. However, since $t' =
    (`l a[\Uparrow(s)][s_1]\cdots[s_w]) \in L(\pi)$ it also means that $a \in
    L(\alpha)$, $s \in L(\sigma)$, and $s_i \in L(\sigma_i)$ for all $1 \leq i
    \leq w$. Consequently, $t \in L((`l
    \alpha)[\sigma][\sigma_1]\cdots[\sigma_w])$ as required.

    Conversely, let $\pi = `l(\alpha[\Uparrow(\sigma)]) [\sigma_1] \cdots
    [\sigma_w]$ be the unique term in the $\varphi$\nobreakdash-matching family
    $\Delta_\varphi^n$ such that $t \in L((`l
    \alpha)[\sigma][\sigma_1]\cdots[\sigma_w])$. Note that $a \in L(\alpha)$, $s
    \in L(\sigma)$, and $s_i \in L(\sigma_i)$ for all $1 \leq i \leq w$.  Since
    $t$ has a head $`y$\nobreakdash-redex, after a single reduction step $t$
    reduces to $t' = `l (a[\Uparrow(s)])[s_1]\cdots[s_w] \in L(\pi)$. By
    construction of $\Delta_\varphi^n$, it
    means that there exists a production $G_n \to \gamma$ in $\mathscr{G}_n$ such that $L(\pi)
    \subset L(\gamma)$ and hence $t' \downarrow_n$. Certainly, it follows that $t
    \downarrow_{n+1}$.
\end{proof}

Let us remark that almost all of the rewriting rules of $`l`y$ exhibit a
similar construction scheme; the exceptional~\eqref{red:app}
and~\eqref{red:fvar} rewriting rules are discussed
in~\cref{rem:app:construction} and~\cref{rem:production:scheme:closure:width},
respectively. Given a rewriting rule, we start with the respective template
$\varphi$ (see~\cref{tab:templates}) and generate all possible
$\varphi$\nobreakdash-matchings in $\mathscr{G}_n$.  Intuitively, such an
operation extracts unambiguous sublanguages out of each production in
$\mathscr{G}_n$ which match the right-hand side of the considered rewriting
rule.  Next, we consider each term $\pi \in \Delta_\varphi^n$ and establish new
productions $G_{n+1} \to \gamma$ in $\mathscr{G}_{n+1}$ out of $\pi$. Assuming
that $\mathscr{G}_n$ is a simple and verbose \urg, the novel productions
generated by means of $\Delta_\varphi^n$ cover all the \luterms reducing in
$n+1$ normal-order steps, starting with the prescribed head rewriting rule.  Since
the head of each so constructed production either starts with a function symbol
or is equal to $N$, cf.~\cref{tab:templates}, the outcome grammar is necessarily
verbose. Moreover, if we complete the production generation for all rewriting
rules, by construction, the grammar $\mathscr{G}_{n+1}$ must be, at the same
time, simple. Consequently, the construction of the hierarchy
$\seq{\mathscr{G}_n}_n$ amounts to an inductive application of the above
construction scheme.

\begin{remark}\label{rem:app:construction}
    While following the same pattern for the~\eqref{red:app} rule, we notice that the
    corresponding construction requires a slight modification. Specifically,
    while matching $\varphi = T[S](T[S])$ with a right-hand side $\gamma$ of a production
    $G_n \to \gamma$ in $\mathscr{G}_n$ we cannot conclude that $\pi \in
    \Delta_\varphi^n$ takes the form $\pi = \alpha [\sigma] (\beta [\sigma])
    [\sigma_1] \cdots [\sigma_w]$. Note that, in fact, $\pi = \alpha [\tau_1]
    (\beta [\tau_2]) [\sigma_1] \cdots [\sigma_w]$ however perhaps
    $\tau_1 \neq \tau_2$. Nonetheless, we can still compute $\Pi(\tau_1,\tau_2)$
    and use $\tau \in \Pi(\tau_1,\tau_2)$ to generate a finite set of terms in
    form of $\pi = \alpha [\tau] (\beta [\tau]) [\sigma_1] \cdots [\sigma_w]$.
    Using those terms, we can continue with our construction and establish a set
    of new productions in
    form of $G_n \to (\alpha \beta) [\tau] [\sigma_1] \cdots [\sigma_w]$
    meant to be included in $\mathscr{G}_{n+1}$.
\end{remark}

\begin{remark}\label{rem:production:scheme:closure:width}
Let us also remark that the single rewriting rule which has a template $\varphi$ not
    retaining closure width is~\eqref{red:fvar}. In consequence, the utility of
    $\Delta_T(\gamma)$ is substantially limited.  If $t =
    \idx{0}[a/][s_1]\cdots[s_w] \downarrow_{n+1}$, then $t \to t' = a
    [s_1]\cdots[s_w]$ which, in turn, satisfies $t' \downarrow_n$. Note that if
    $\gamma$ is the right-hand side of a unique production $G_n \to \gamma$ in
    $\mathscr{G}_n$ generating $t'$, then we can match $T$ with \emph{any}
    non-empty prefix of $\gamma$. The length of the chosen prefix influences
    what initial part $\alpha$ of $\gamma$ is going to be placed under the
    closure in $G_{n+1} \to \idx{0}[\alpha/][\sigma_1]\cdots[\sigma_w]$.

    This motivates the following approach. Let $G_n \to \gamma' =
    \chi[\sigma_1]\cdots[\sigma_w]$ be an arbitrary production in
    $\mathscr{G}_n$ of closure width $w$. If $t' \in L(\gamma')$ and $t \to t'$
    in a single head~\eqref{red:fvar}\nobreakdash-reduction, then $t \in
    L(\idx{0}[\chi[\sigma_1]\cdots[\sigma_d]/][\sigma_{d+1}]\cdots[\sigma_w]$
    for some $0 \leq d \leq w$. Therefore, in order to generate all productions
    in $\mathscr{G}_{n+1}$ corresponding to \luterms
    $`y$\nobreakdash-normalising in $n+1$ steps, starting with a
    head~\eqref{red:fvar}\nobreakdash-reduction, we have to include all
    productions in form of $G_{n+1} \to
    \idx{0}[\chi[\sigma_1]\cdots[\sigma_d]/][\sigma_{d+1}]\cdots[\sigma_w]$ for
    each production $G_n \to \gamma'$ in $\mathscr{G}_n$.

    Finally, note that it is, again, possible to optimise
    the~\eqref{red:fvar} construction scheme with respect to the number of
    generated productions. For each $G_n \to \gamma$ in $\mathscr{G}_n$ the
    above scheme produces, \emph{inter alia}, a production $G_{n+1} \to
    \idx{0}[\gamma/]$. Note that we can easily merge them into a single production
    $G_{n+1} \to \idx{0}[G_n/]$ instead.
\end{remark}

Such a construction leads us to the following conclusion.

\begin{theorem}\label{th:main:result}
    For all $n \geq 0$ there exists a constructible, simple \urg $\mathscr{G}_n$.
\end{theorem}

\begin{example}\label{ex:g1}
The following example demonstrates the construction of $\mathscr{G}_1$ out of
    $\mathscr{G}_0$.  Note that $\mathscr{G}_1$ includes the following
    productions associated with the axiom $G_1$:
    \begin{align}
    \begin{split}
        G_1 &\to `l G_1~|~G_0 G_1~|~G_1 G_0\\
            &\quad|~\idx{0} [(G_0 G_0) /]~|~\idx{0}[`l G_0/]~|~\idx{0}[N/]\\
            &\quad|~(\succ N) [T/]~|~\idx{0}[\Uparrow(S)]~|~N[\uparrow].
    \end{split}
\end{align}

    The first three productions are included by default. The next three productions
    are derived from the~\eqref{red:fvar} rule applied to all the productions of
    $G_0 \to \gamma$ in $\mathscr{G}_0$. The final three productions are
    obtained by~\eqref{red:rvar},~\eqref{red:fvarlift},
    and~\eqref{red:varshift}, respectively.
\end{example}

\section{Analytic combinatorics and simple $`y$-reduction
grammars}\label{sec:analytic:combinatorics}

Having established an effective hierarchy $\seq{\mathscr{G}_k}_k$ of simple
$`y$\nobreakdash-reduction grammars, we can now proceed with their quantitative
analysis. Given the fact that regular tree grammars represent well-known
algebraic tree-like structures, our analysis is in fact a standard application
of algebraic singularity analysis of respective generating
functions~\cite{FlajoletOdlyzko1990,flajolet09}. The following result provides
the main tool of the current section.

\begin{prop}[Algebraic singularity analysis, see~\cite{flajolet09}, Theorem
    VII.8]\label{prop:singularity-analysis}
    Assume that $f(z) = \left(\sqrt{1 - z/\zeta}\right) g(z) + h(z)$ is an
    algebraic function, analytic at $0$, and has a unique dominant singularity
    $z = \zeta$.  Moreover, assume that $g(z)$ and $h(z)$ are analytic in the
    disk $|z| < \zeta + \varepsilon$ for some $\varepsilon > 0$. Then, the
    coefficients $[z^n]f(z)$ in the series expansion of $f(z)$ around the
    origin, satisfy the following asymptotic estimate
\begin{equation}
    [z^n]f(z) \xrightarrow[n \to \infty]{} \zeta^{-n} \frac{
        n^{-3/2} g(\zeta)}{\Gamma(-\frac{1}{2})}.
\end{equation}
\end{prop}

In order to analyse the number of \luterms normalising in $k$ steps, we execute
the following plan. First, we use the structure (and unambiguity) of
$\mathscr{G}_k$ to convert it by means of symbolic methods into a corresponding
generating function $G_k(z) = \sum g^{(k)}_n z^n$ in which the integer
coefficient $g^{(k)}_n$ standing by $z^n$ in the series expansion of $G_k(z)$,
also denoted as $[z^n]G_k(z)$, is equal to the number of \luterms of size $n$
normalising in $k$ steps. Next, we show that so obtained generating functions
fit the premises of~\cref{prop:singularity-analysis}.

We start with establishing an appropriate size notion for \luterms.  For
technical convenience, we assume the following
\emph{natural size notion}, equivalent to the number of constructors in the associated
term algebra $\mathscr{T}_\mathscr{F}(\emptyset)$, see~\cref{fig:size:notion}.
\begin{figure}[hbt!]
    \centering
\noindent\begin{minipage}{.3\linewidth}
\begin{align*}
    |\idx{n}| &= n + 1\\
    |`l a| &= 1 + |a|\\
    |a b| &= 1 + |a| + |b|\\
    |a[s]| &= 1 + |a| + |s|
\end{align*}
\end{minipage}%
\begin{minipage}{.3\linewidth}
\begin{align*}
    |a/| &= 1 + |a|\\
    |\Uparrow(s)| &= 1 + |s|\\
    |\uparrow| &= 1.
\end{align*}
\end{minipage}%
    \caption{Natural size notion for \luterms.}
    \label{fig:size:notion}
\end{figure}

The following results exhibit the closed-form of generating functions
corresponding to pure terms as well as the general class of \luterms and
explicit substitutions.

\begin{prop}[see~\cite{doi:10.1093/logcom/exx018}]\label{prop:Linfty}
    Let $L_\infty(z)$ denote the generating function corresponding to
    the set of \lterms in $`y$\nobreakdash-normal form (i.e.~without
    $`y$\nobreakdash-redexes). Then,
    \begin{equation}
    L_\infty(z) = \frac{1-z-\sqrt{\frac{1 - 3z - z^2- z^3}{1-z}}}{2 z}.
    \end{equation}
\end{prop}

\begin{prop}[see~\cite{Bendkowski:2018:CES:3236950.3236951}]\label{prop:basic:gen:fun}
    Let $T(z)$, $S(z)$ and $N(z)$ denote the generating functions corresponding to
    \luterms, substitutions, and de~Bruijn indices, respectively. Then,
    \begin{equation}\label{eq:plain:terms:genfun}
      T(z) = \dfrac{1-\sqrt{1-4z}}{2z} - 1, \quad
      S(z) = \frac{1-\sqrt{1-4z}}{2z} \left(\frac{z}{1-z}\right)
        \quad \text{and} \quad N(z) = \frac{z}{1-z}.
    \end{equation}
\end{prop}

With the above basic generating functions, we can now proceed with the
construction of generating functions corresponding to simple
$`y$\nobreakdash-reduction grammars.

\begin{prop}[Constructible generating functions]\label{prop:gn:form}
    Let $\Phi_k$ denote the set of regular productions in $\mathscr{G}_k$.
    Then, for all $k \geq 1$ there exists a generating function $G_k(z)$
    such that $[z^n]G_k(z)$ (i.e.~the coefficient standing by $z^n$ in the power
    series expansion of $G_k(z)$) is equal to the number of terms of size
    $n$ which $`y$\nobreakdash-normalise in $k$ normal-order reduction steps,
    and moreover $G_k(z)$ admits a closed-form of the following shape:
    \begin{equation}\label{eq:Gn:genfun}
        G_k(z) = \dfrac{1}{1-z-2 z L_\infty(z)} \sum_{G_k \to \gamma \in
        \Phi_k} G_\gamma(z)
    \end{equation}
    where
    \begin{equation}\label{eq:regular:production:genfun}
        G_\gamma(z) = z^{\zeta(\gamma)} {T(z)}^{\tau(\gamma)}
        {S(z)}^{\sigma(\gamma)} {N(z)}^{\nu(\gamma)} \prod_{0 \leq i < k}
        {G_i(z)}^{\rho_i(\gamma)}
    \end{equation}
    and all $\zeta(\gamma)$, $\tau(\gamma)$, $\sigma(\gamma)$, $\nu(\gamma)$,
    and $\rho_i(\gamma)$ are non-negative integers depending on $\gamma$.
\end{prop}

\begin{proof}
    Let $G_k \to \gamma$ be a regular production in $\mathscr{G}_k$. Since by
    construction $\mathscr{G}_k$ is simple, we know that $\gamma \in
    \mathscr{T}_{\mathscr{F}}(\mathscr{N} \cup \set{G_0,\ldots,G_{k-1}})$.
    Following symbolic methods~\cite[Part A, Symbolic
    Methods]{flajolet09} we can therefore convert each non-terminal $X \in
    \mathscr{N} \cup \set{G_0,\ldots,G_{k-1}}$ occurring in $\gamma$ into an
    appropriate generating function $X(z)$. Likewise, we can convert each
    function symbol occurrence $f$ into an appropriate monomial $z$,
    see~\cref{fig:size:notion}. Finally, we group respective monomials together,
    and note that the generating function $G_\gamma(z)$ corresponding to $\gamma$
    takes the form~\eqref{eq:regular:production:genfun}. Respective
    exponents denote the number of occurrences of their associated symbols.

    Consider the remaining self-referencing productions $G_k \to \delta$.
    Again, since $\mathscr{G}_k$ is simple, we know that $\delta$ takes the form
    $`l G_k$, $G_0 G_k$ or (symmetrically) $G_k G_0$. And so, as each $X \in
    \mathscr{N}_n$ is unambiguous in $\mathscr{G}_k$, by symbolic methods, it
    follows that $G_k(z)$ satisfies the following functional equation:
    \begin{equation}\label{eq:Gn:genfun:ii}
    G_k(z) = z G_k(z) + 2 z G_0(z) G_k(z) + \sum_{G_k \to \gamma \in
        \Phi_k} G_\gamma(z).
    \end{equation}
    Note that as no $G_\gamma(z)$ references the left-hand side $G_k(z)$,
    equation~\eqref{eq:Gn:genfun:ii} is in fact linear in $G_k(z)$. Furthermore,
    as  $G_0(z) = L_\infty(z)$ we finally obtain the requested form of $G_k(z)$,
    see~\eqref{eq:Gn:genfun}.
\end{proof}

This brings us to our following, main quantitative result.

\begin{theorem}\label{th:main:result:asymptotics}
    For all $k \geq 1$ the coefficients $[z^n]G_k(z)$ admit the following
    asymptotic estimate
    \begin{equation}
        [z^n]G_k(z) \xrightarrow[n \to \infty]{} c_k \cdot 4^n n^{-3/2}.
    \end{equation}
\end{theorem}

\begin{proof}
    We claim that for each $k \geq 1$ the generating function $G_k(z)$ can be
    represented as $G_k(z) = \sqrt{1-4z} P(z) + Q(z)$ where both $P(z)$ and
    $Q(z)$ are functions analytic in the disk $|z| < \frac{1}{4} + \varepsilon$
    for some positive $\varepsilon$. The asserted asymptotic estimate follows
    then as a straightforward application of algebraic singularity analysis,
    see~\cref{prop:singularity-analysis}.

    We start with showing that each $G_k(z)$ includes a summand in form of
    $\sqrt{1-4z}\, \overline{P}(z) + \overline{Q}(z)$ such that both
    $\overline{P}(z)$ and $\overline{Q}(z)$ are analytic in a large enough disk
    containing (properly) $|z| < \frac{1}{4}$. Afterwards, we argue that no summand
    has singularities in $|z| < \frac{1}{4}$. Standard closure properties of analytic
    functions with single dominant, square-root type singularities guarantee the
    required representation of $G_k(z)$.

    Let $\varphi = N$ be the template corresponding to the~\eqref{red:rvar} rule,
    see~\cref{tab:templates}. Note that since $\mathscr{G}_0$ includes the
    production $G_0 \to N$, the set of $\varphi$\nobreakdash-matchings
    $\Delta_\varphi^0$ consists of the single term $N$. Hence, due to the
    respective production construction, it means that $G_1 \to (\succ N)[T/]$ is
    a production of $\mathscr{G}_1$, cf.~\cref{ex:g1}.  Moreover, as a consequence of the
    construction associated with the~\eqref{red:fvar} rule, for each $G_k \to
    \gamma$ in $\mathscr{G}_k$ there exists a production $G_{k+1} \to
    \idx{0}[\gamma/]$ in the subsequent grammar $\mathscr{G}_{k+1}$. And so, each
    $\mathscr{G}_{k}$ includes among its productions one production in form of
    \begin{equation}
        G_{k} \to \underbrace{\idx{0}[\idx{0}[\ldots\idx{0}}_{k-1 \text{
            times}}[(\succ N)[T/]/]\ldots/]/]
    \end{equation}
    Denote the right-hand side of the above production as $\gamma$.  Note that
    the associated generating function $G_\gamma(z)$,
    cf.~\eqref{eq:Gn:genfun:ii}, must therefore take form
    \begin{equation}\label{eq:major:production}
        G_\gamma(z) = z^{\zeta(\gamma)} T(z) N(z) \qquad \text{and so} \qquad
        G_\gamma(z) = \sqrt{1-4z}\, \overline{P}(z) + \overline{Q}(z)
    \end{equation}
    where both $P(z)$ and $Q(z)$ are analytic in $|z| < \frac{1}{4} +
    \varepsilon$ for some (determined) $\varepsilon > 0$.

    In order to show that no production admits a corresponding generating
    function with singularities in the disk $|z| < \frac{1}{4}$ we note that the
    single dominant singularity of $G_0(z)$, and so at the same time
    $L_\infty(z)$, is equal to the smallest positive real root $\rho$ of
    $1-3z-z^2-z^3$ which satisfies $\frac{1}{4} < \rho \approx 0.295598$,
    see~\cref{prop:Linfty}. Due to the form of basic generating functions
    corresponding to $T$, $S$ and $N$, see~\cref{prop:basic:gen:fun}, we further
    note that other singularities must lie on the unit circle $|z| = 1$.  And
    so, each $G_k(z)$ admits the asserted form
    \begin{equation}
    G_k(z) = \sqrt{1-4z} P(z) + Q(z)
    \end{equation}
    for some functions analytic in a disk $|z| < \frac{1}{4} + \varepsilon$.
\end{proof}

\begin{remark}
    Consider the following \emph{asymptotic density} of \luterms
    $`y$\nobreakdash-normalisable in $k$ normal-order reduction steps
    in the set of all \luterms:
    \begin{equation}\label{eq:asymp:density}
        \mu_k = \lim_{n\to\infty} \dfrac{[z^n]G_k(z)}{[z^n]T(z)}.
    \end{equation}
    In other words, the limit $\mu_k$ of the probability that a uniformly random
    \luterm of size $n$ normalises in $k$ steps as $n$ tends to infinity.  Note
    that for each $k \geq 1$, the asymptotic density $\mu_k$ is positive as both
    $[z^n]G_k(z)$ and $[z^n]T(z)$ admit the same (up to a multiplicative
    constant) asymptotic estimate.  Moreover, it holds $\mu_k \xrightarrow[k \to
    \infty]{} 0$ as the sum $\sum_k \mu_k$ is increasing and necessarily bounded
    above by one.

    Each \luterm is $`y$\nobreakdash-normalising in some (determined) number of
    normal-order reduction steps. However, it is not clear whether $\sum_k \mu_k
    = 1$ as asymptotic density is, in general, not countably additive.  Let us
    remark that if this sum converges to one, then the random variable $X_n$
    denoting the number of normal-order $`y$\nobreakdash-reduction steps
    required to normalise a random \luterm of size $n$ (i.e.~the average-case
    cost of resolving pending substitutions in a random term of size $n$)
    converges pointwise to a discrete random variable $X$ defined as
    $\mathbb{P}(X = k) = \mu_k$. Unfortunately, the presented analysis does not
    determine an answer to this problem.
\end{remark}

\begin{remark}
    \cref{th:main:result} and~\cref{th:main:result:asymptotics} are
    effective in the sense that both the symbolic representation and the
    symbolic asymptotic estimate of respective coefficients are computable.
    Since $\Gamma(-1/2) = -2 \sqrt{\pi}$ is the sole transcendental number
    occurring in the asymptotic estimates, and cancels out when asymptotic
    densities are considered, cf.~\eqref{eq:asymp:density}, we immediately note that for each $k \geq 0$, the
    asymptotic density of terms $`y$\nobreakdash-normalising in $k$ steps is
    necessarily an algebraic number.
\end{remark}

    \cref{fig:densities:lucalculus} provides the specific asymptotic densities
    $\mu_0,\ldots,\mu_{10}$ obtained by means of a direct construction and
    symbolic computations\footnote{Corresponding software is
    available at~\url{https://github.com/maciej-bendkowski/towards-acasrlc}.}
    (numerical values are rounded up to the fifth decimal point).
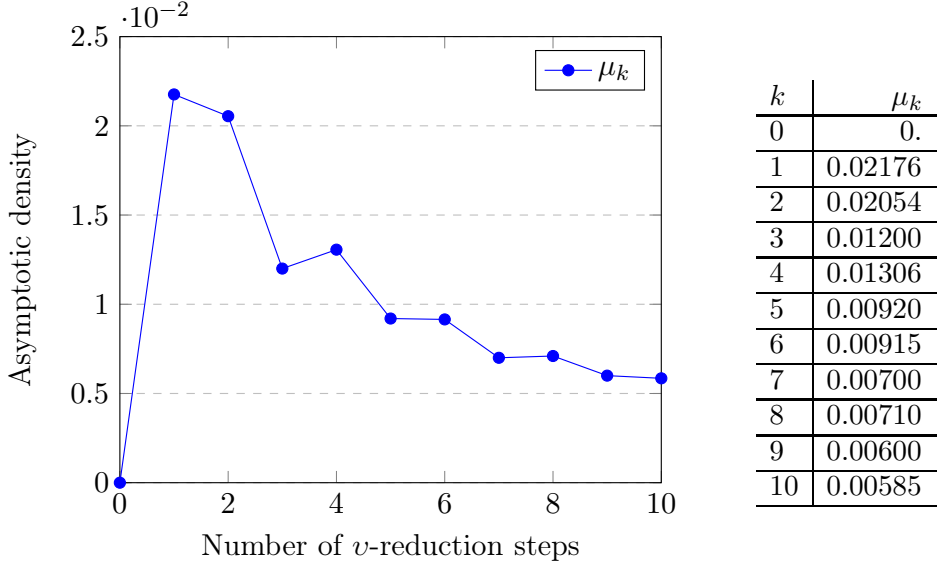
\begin{figure}[th!]
\centering
\begin{subfigure}{.55\textwidth}
\centering
      \resizebox {260px} {!} {
    \begin{tikzpicture}
\begin{axis}[
    xlabel={Number of $`y$-reduction steps},
    ylabel={Asymptotic density},
    xmin=0, xmax=10,
    ymin=0, ymax=0.025,
    legend pos=north east,
    ymajorgrids=true,
    grid style=dashed,
]

\addplot[
    color=blue,
    mark=*,
    ]
    coordinates {
        (0,0)
        (1,0.02176)
        (2,0.02054)
        (3,0.01200)
        (4,0.01306)
        (5,0.00920)
        (6,0.00915)
        (7,0.00700)
        (8,0.00710)
        (9,0.00600)
        (10,0.00585)
    };
    \legend{$\mu_k$}

\end{axis}
\end{tikzpicture}}
\end{subfigure}
\begin{subfigure}{.3\textwidth}
\centering
\begin{displaymath}
        \begin{array}{l|r}
        k & \mu_k  \\\hline
        0 & 0. \\\hline
        1 & 0.02176  \\\hline
        2 & 0.02054 \\\hline
        3 & 0.01200 \\\hline
        4 & 0.01306 \\\hline
        5 & 0.00920 \\\hline
        6 & 0.00915 \\\hline
        7 & 0.00700 \\\hline
        8 & 0.00710 \\\hline
        9 & 0.00600 \\\hline
        10 & 0.00585 \\\hline
        \end{array}
\end{displaymath}
\end{subfigure}
\caption{Asymptotic densities of terms $`y$\nobreakdash-normalising in $k$
    normal-order reduction steps. In particular, we have
    $\mu_0 + \cdots + \mu_{10} \approx 0.11162$.}\label{fig:densities:lucalculus}
\end{figure}

\section{Applications to other term rewriting
systems}\label{sec:other:applications}
Let us note that the presented construction of reduction grammars does not
depend on specific traits immanent to $`l`y$, but rather on certain more general
features of its rewriting rules. The key ingredient in our approach is the
ability to compute finite intersection partitions for arbitrary terms within the
scope of established reduction grammars, which themselves admit certain neat
structural properties.  Using finite intersection partitions, it becomes
possible to generate new productions based on the structural resemblance of both
the left-hand and right-hand sides of associated rewriting rules.

In what follows we sketch the application of the presented approach to other
term rewriting systems, focusing on two examples of \lucalculus and combinatory
logic. Although the technique is similar, some careful amount of details is
required.

\subsection{$`l`y$-calculus}
In order to characterise the full $`l`y$\nobreakdash-calculus, we need a few
adjustments to the already established construction of $\seq{\mathscr{G}_k}_k$
corresponding to the $`y$ fragment.
Clearly, we have to establish a new
production construction scheme associated with the~\eqref{red:beta} rewriting
rule. Consider the corresponding template $\varphi = T[T/]$. Like the
respective template for~\eqref{red:fvar},
cf.~\cref{rem:production:scheme:closure:width}, the current template $\varphi$
does not retain closure width of generated ground terms. However, at the same
time it is also amenable to similar treatment as~\eqref{red:fvar}.

Let $t = (`l a)b[s_1]\cdots[s_w] \downarrow_{n+1}$. Note that $t \to
a[b/][s_1]\cdots[s_w] \downarrow_n$. Consequently, one should attempt to match
$\varphi$ with all possible \emph{prefixes} of $\gamma$ in all productions $G_n
\to \gamma$ instead of merely their heads, as in the case of $\Delta_\varphi^n$.
Let $\gamma = \chi[\sigma_1]\cdots[\sigma_d]$ be of closure width $d$ and
$\delta = \chi[\sigma_1]\cdots[\sigma_i]$ be its prefix ($1 \leq i \leq d$).
Note that, effectively, $\Pi(\delta,T[T/])$ checks if $[\sigma_i]$ takes form
$\beta/$ for some term $\beta$. Hence, for each $\pi \in \Pi(\delta,T[T/])$ we
have $\pi = \alpha[\tau_1]\cdots[\tau_{i-1}][\beta/]$ for some terms
$\alpha$, $\tau_1,\ldots,\tau_{i-1}$, and $\beta$. Out of such a partition $\pi$
we can then construct the production $G_{n+1} \to (`l
\alpha[\tau_1]\cdots[\tau_{i-1}])\beta[\tau_{i+1}]\cdots[\tau_d]$ in
$\mathscr{G}_{n+1}$.

However, with the new scheme for~\eqref{red:beta} we are not yet done with our
construction.  Note that due to the new head redex type, we are, \emph{inter alia},
forced to change the pre-defined set of productions corresponding to \luterms
without a head redex.  Indeed, note that if $t$ does not start with a head redex, then
it must be of form $(`l a)$ or $(a b)$ where in the latter case $a$ cannot start
with a leading abstraction. This constraint suggests the following approach. We
split the set of productions in form of $G_n \to \gamma$ into two categories,
i.e.~productions whose right-hand sides are of form $`l \gamma'$ and the
reminder ones. Since both sets are disjoint, we can group them separately and
introduce two auxiliary non-terminal symbols $G_n^{(`l)}$ and $G_n^{(\neg `l)}$
for terms starting with a head abstraction and for those without a head
abstraction, respectively, with the additional productions $G_n \to
G_n^{(`l)}~|~G_n^{(\neg `l)}$. In doing so, it is possible to pre-define the all
the productions corresponding to terms without head redexes using productions in
form of
\begin{equation}
    G_n \to `l G_n~|~G_k^{(\neg `l)}~G_{n-k} \qquad \text{where} \qquad 0 \leq k \leq
    n.
\end{equation}
This operation, however, requires a minor adjustment of the term potential,
see~\cref{def:term:potential}, accounting for the new non-terminal symbols
$G_k^{(`l)}$ and $G_k^{(\neg `l)}$ occurring in the right-hand sides of grammar
productions. Once updated, we can reuse it in showing that finite intersection
partitions are, again, computable and can serve as a tool in the construction of
new productions in $\mathscr{G}_{n+1}$ out of $\mathscr{G}_n$.

\subsection{$S K$-combinators}
Using a similar approach, it becomes possible to construct reduction grammars
for $SK$\nobreakdash-combinators. In particular, our current technique
(partially) subsumes, and also simplifies, the results
of~\cite{bengryzai2017,BENDKOWSKI_2017}. With merely two rewriting rules in form
of
\begin{align}
\begin{split}
    K x y &\to x\\
    S x y z &\to x z (y z)
\end{split}
\end{align}
we can use the developed finite intersection partitions and
$\varphi$\nobreakdash-matchings to construct a hierarchy $\seq{\mathscr{G}_n}_n$
of normal-order reduction grammars for $S K$\nobreakdash-combinators. The
rewriting rule corresponding to $K$ is similar to~\eqref{red:fvar} whereas the
respective rule for $S$ resembles the~\eqref{red:app} rule; as in this case, we have
to deal with variable duplication on the right-hand side of the rewriting rule.
Instead of closure width, we use a different normal form of terms, and so also
productions, based on the sole binary constructor of term
application. Consequently, a combinator is of \emph{application width} $w$ if it
takes the form $X \alpha_1 \ldots \alpha_w$ for some primitive combinator $X \in
\set{S,K}$.

Consider the more involved case of productions corresponding to head
$S$\nobreakdash-redexes. Let $t = S x y z \alpha_1 \ldots \alpha_w$ be a term
of application width $w + 3$ where $w \geq 0$. Note that
\begin{equation}\label{eq:S:combinator}
    S x y z \delta_1 \ldots \delta_w \to x z (y z) \delta_1 \ldots \delta_w.
\end{equation}
Let us rewrite the right-hand side of~\eqref{eq:S:combinator} as
$t' = X x_1 \ldots x_k z (y z) \delta_1 \ldots \delta_w$ where $x = X x_1 \ldots
x_k$ and $X$ is a primitive combinator. Assume that $\gamma$ is the right-hand
side of the unique production $G_n \to \gamma$ in $\mathscr{G}_n$ such that $t'
\in L(\gamma)$. Note that the shape of $t'$ suggests a
construction scheme similar to the already discussed~\eqref{red:app},
see~\cref{rem:app:construction}, where we first have to match the pattern
$\varphi = T (T T)$ with some argument of $\gamma$ and subsequently attempt to
extract a finite intersection partition $\Pi(\alpha,\beta)$ of respective subterms
$\alpha$ and $\beta$ so that for each $\pi \in \Pi(\alpha,\beta)$ we have $z \in
L(\pi)$. With appropriate terms at hand, we can then construct corresponding
productions in the next grammar $\mathscr{G}_{n+1}$.

\section{Conclusions}\label{sec:conclusion}
Quantitative aspects of term rewriting systems are not well studied. A general
complexity analysis was undertaken by Choppy, Kaplan, and Soria~who considered a
class of confluent, terminating term rewriting systems in which the evaluation
cost, measured in the number of rewriting steps required to reach the normal
form, is independent of the assumed evaluation
strategy~\cite{DBLP:journals/tcs/ChoppyKS89}. More recently, typical evaluation
cost of normal-order reduction in combinatory logic was studied by Bendkowski,
Grygiel and Zaionc~\cite{bengryzai2017,BENDKOWSKI_2017}. Using quite different,
non-analytic methods, Sin’Ya, Asada, Kobayashi and Tsukada considered certain
asymptotic aspects of $`b$\nobreakdash-reduction in the simply-typed variant of
\lcalculus showing that, typically, \lterms have exponentially long
$`b$\nobreakdash-reduction sequences~\cite{SinYa:2017:AST:3080372.3080377}.

Arguably, the main novelty in the presented approach lies in the algorithmic
construction of reduction grammars $\seq{\mathscr{G}_k}_k$ based on finite
intersection partitions, assembled using a general technique reminiscent of
Robinson's unification algorithm applied to many-sorted term algebras,
cf.~\cite{Robinson:1965:MLB:321250.321253,Walther:1988}. Equipped with finite
intersection partitions, the construction of $\mathscr{G}_{k+1}$ out of
$\mathscr{G}_k$ follows a stepwise approach, in which new productions are
established on a per rewriting rule basis. Consequently, the general technique
of generating reduction grammars does not depend on specific features of $`l`y$,
but rather on more general traits of certain first-order rewriting systems.
Nonetheless, the full scope of our technique is yet to be determined.

Although the presented construction is based on the leftmost-outermost
evaluation strategy, it does not depend on the specific size notion associated
with \luterms; in principle, more involved size models can be assumed and
analysed. The assumed evaluation strategy, size notion, as well as the
specific choice of $`l`y$ are clearly arbitrary and other, equally perfect choices
for modelling substitution resolution could have been made. However, due to
merely eight rewriting rules forming $`l`y$, it is one of the conceptually
simplest calculus of explicit substitutions. Together with the normal-order
evaluation tactic, it is therefore one of the simplest to investigate in
quantitative terms and to demonstrate the finite
intersection partitions technique.

Due to the unambiguity of constructed grammars $\seq{\mathscr{G}_k}_k$ it is
possible to automatically establish their corresponding combinatorial
specifications and, in consequence, obtain respective generating functions
encoding sequences $\seq{g^{(k)}_{n}}_n$ comprised of numbers $g^{(k)}_{n}$ associated
with \luterms of size $n$ which reduce in $k$ normal-order rewriting steps to
their $`y$\nobreakdash-normal forms. Singularity analysis provides then the
means for systematic, quantitative investigations into the properties of
substitution resolution in $`l`y$, as well as its machine-independent
operational complexity. Finally, with generating functions at hand, it is
possible to undertake a more sophisticated statistical analysis of substitution
(in particular $`y$\nobreakdash-normalisation) using available techniques of
analytic combinatorics, effectively analysing the average-case cost of
\lcalculus and related term rewriting systems.

\bibliographystyle{plain}
\bibliography{references}

\end{document}